\documentclass[preprint,3p]{elsarticle}

\usepackage[T1]{fontenc}
\usepackage{graphicx, algorithm, algorithmic}
\usepackage[utf8]{inputenc}
\usepackage{tikz}
\usetikzlibrary{shapes.geometric,shapes.misc,shapes.multipart}
\usetikzlibrary{calc,shapes,arrows,arrows.meta,shapes.misc,decorations.pathreplacing}
\usetikzlibrary{positioning,calc}
\usetikzlibrary{automata, positioning}
\usepackage{microtype}

\usepackage{tikz-qtree}
\usepackage{subcaption}
\usepackage{lipsum}
\usepackage{listings}
\usepackage{adjustbox}
\usepackage{amsmath,amssymb}
\usepackage{hyperref}
\usepackage{ textcomp }
\usepackage{ wasysym }
\usepackage{bm}
\usepackage{xcolor}

\tikzstyle{nodesty}=[circle,draw,minimum size=0.5cm,inner sep=0.5pt]
\tikzset{>={triangle 45}}

\title{Scattered one-counter languges have rank less than $\omega^2$}
\author[1]{Szabolcs~Iv\'an}
\ead{szabivan@inf.u-szeged.hu}
\address[1]{Department of Computer Science, University of Szeged, Hungary\\ }
	
\def\pref{{\mathbf{Pref}}}

\newtheorem{proposition}{Proposition}
\newtheorem{corollary}{Corollary}
\newtheorem{theorem}{Theorem}
\newproof{proof}{Proof}
\newtheorem{definition}{Definition}

\begin{document}
\begin{abstract}
	A linear ordering is called context-free if it is the lexicographic ordering of some context-free language
	and is called scattered if it has no dense subordering. Each scattered ordering has an associated ordinal,
	called its rank. It is known that scattered context-free (regular, resp.) orderings have rank less than $\omega^\omega$
	($\omega$, resp).
	
	A language is called a one-counter language if it can be recognized by a pushdown automaton having a singleton
	stack alphabet (essentially working as a counter that can hold a nonnegative integer and can be tested
	against zero). The class of one-counter languages lies strictly between the classes of the regular
	and the context-free languages.
	
	In this paper we confirm the conjecture of Dietrich Kuske from 2012
	that scattered one-counter languges have rank less than $\omega^2$.
\end{abstract}

\maketitle
\section{Introduction}
If an alphabet $\Sigma$ is equipped by a linear order $<$, this order can be extended to the lexicographic
ordering $<_\ell$ on $\Sigma^*$ as $u<_\ell v$ if and only if either $u$ is a proper prefix of $v$ or
$u=xay$ and $v=xbz$ for some $x,y,z\in\Sigma^*$ and letters $a<b$. So any language $L\subseteq \Sigma^*$
can be viewed as a linear ordering $(L,<_\ell)$. Since $\{a,b\}^*$ contains the dense ordering
$(aa+bb)^*ab$ and every countable linear ordering can be embedded into any countably infinite dense ordering,
every countable linear ordering is isomorphic to one of the form $(L,<_\ell)$ for some language $L\subseteq\{a,b\}^*$.

This way, countable order types can be represented by languages over some alphabet
(by a prefix-free encoding of the alphabet
by binary strings, one can restrict the alphabet to the binary one). A very natural choice is to use regular
or context-free languages as these language classes are well-studied.
A linear ordering (or an order type) is called \emph{regular} or \emph{context-free}
if it is isomorphic to the linear ordering (or, is the order type) of some language of the appropriate class.
It is known~\cite{DBLP:journals/fuin/BloomE10} that an ordinal is regular if and only if it is less than $\omega^\omega$
and is context-free if and only if it is less than $\omega^{\omega^\omega}$. Also, the Hausdorff rank~\cite{rosenstein}
of any scattered regular (context-free, resp.) ordering is less than $\omega$ ($\omega^\omega$, resp)~\cite{ITA_1980__14_2_131_0,10.1007/978-3-642-29344-3_25}.

It is known~\cite{GelleIvanTCS} that the order type of a well-ordered language generated by a prefix grammar (i.e. in which
each nonterminal generates a prefix-free language)
is computable,
thus the isomorphism problem of context-free ordinals is decidable if the ordinals in question are given as the lexicograpic
ordering of \emph{prefix} grammars.
Also, the isomorphism problem of regular orderings is decidable as well~\cite{DBLP:journals/ita/Thomas86,BLOOM200555},
even in polynomial time~\cite{LOHREY201371}.
On the other hand, it is undecidable for a context-free grammar whether it generates a dense language,
hence the isomorphism problem of context-free orderings in general is undecidable~\cite{ESIK2011107}.
It is unknown whether the isomorphism problem of scattered context-free orderings is decidable --
a partial result in this direction is that if the rank of such an ordering is at most one (that is,
the order type is a finite sum of the terms $\omega$, $-\omega$ and $1$), then the order type is
effectively computable from a context-free grammar generating the language~\cite{GelleIvanGandalf,sofsem2020}. Also,
it is also decidable whether a context-free grammar generates a scattered language of rank at most one.

It is a very plausible scenario though that the isomorphism problem of scattered context-free orderings is
undecidable in general -- the rank $1$ is quite low compared to the upper bound $\omega^\omega$ of the rank
of these orderings, and there is no known structural characterization of scattered context-free orderings.
Clearly, among the well-orderings, exactly the ordinals smaller than $\omega^{\omega^\omega}$ are context-free
but for scattered orderings the main obstacle is the lack of a finite ``normal form'' -- as every $\omega$-indexed
sum of the terms $\omega$ and $-\omega$ is scattered of rank two, and these order types are pairwise different,
there are already uncountably many scattered
orderings of rank two and thus only a really small fraction of them can possibly be context-free.

The class of the one-counter languages lies strictly between the classes of regular and context-free languages:
these are the ones that can be recognized by a pushdown automaton having only one stack symbol.
In~\cite{kuske}, a family of well-ordered languages $L_n\subseteq\{a,b,c\}^*$ was given for each integer $n\geq 0$
so that the order type of $L_n$ is $\omega^{\omega\times n}$ (thus its rank is $\omega\times n$) and
Kuske formulated two conjectures: i) the order type of well-ordered one-counter languages is strictly less
than $\omega^{\omega^2}$ and more generally, ii) the rank of scattered one-counter languages is strictly
less than $\omega^2$. Of course the second conjecture implies the first.

In this paper we prove the second conjecture of~\cite{kuske}: $\omega^2$ is a strict upper bound for the
rank of scattered one-counter languages. The contents of the paper contain new results only: instead of
reproving the results of~\cite{GelleIvanGandalf} and the subsequent, more general~\cite{sofsem2020}
(these papers already contain full proofs and examples as well to their respective results),
we push the boundaries of the knowledge of scattered context-free orderings by applying some of the
tools we developed in the earlier papers to the class of one-counter languages. It turns out that
it is enough to study restricted one-counter languages to prove the conjecture, and for this,
a crucial step is to reason about the cycles in a generalized sequential machine -- so at the end,
we can again use some graph-theoretic methods.
\section{Notation}	
	We assume the reader has some background with formal language theory and linear orderings
	(the textbooks \cite{Hopcroft+Ullman/79/Introduction,rosenstein} being excellent resources for that),
	but we list the notions we use in the paper to settle the notation.
	\subsection{Linear orderings}
	A (strict) \emph{linear ordering} is a pair $(A,<)$ with $A$ being a set, the \emph{domain} of the ordering and
	$<$ being a binary relation over $A$ which is \emph{irreflexive}: $x\not<x$, \emph{transitive}: $x<y,y<z~\Rightarrow x<z$ and \emph{total}: for any $x,y\in A$, exactly one of $x<y$, $y<x$ or $x=y$ holds. In particular,
	the empty set equipped with the empty ordering relation is also a linear ordering. To ease notation, we
	sometimes write $I=(I,<)$ and denote the linear ordering $(I,<)$ simply by its domain $I$ when there is
	no chance of confusion.
	
	The linear ordering $(A,<_A)$ is a \emph{subordering} of the linear ordering $(B,<_B)$ if $A\subseteq B$ and
	$<_A$ is the restriction of $<_B$ to $A$. The linear ordering $(A,<_A)$ can be \emph{embedded} into $(B,<_B)$
	if there is a mapping $h:A\to B$ that \emph{preserves order} (if $x<y$, then $h(x)<h(y)$) -- such mappings
	are called \emph{embeddings} of $A$ into $B$. Clearly, embeddings are injective ($h(x)=h(y)$ implies $x=y$);
	if an embedding is also onto $B$ (also said ``surjective'': for each $x\in B$ there is some $x'\in A$ with
	$h(x')=x$), then $h$ is called an $\emph{(order) isomorphism}$ between $A$ and $B$. If there exists an
	isomorphism between two linear orderings, then we call them \emph{isomorphic}, denoted $A\simeq B$.
	
	Isomorphism of linear orderings is an equivalence relation on any set of linear orderings, an \emph{order type}
	is an equivalence class of isomorphism. Of course, for any integer $n\geq 0$, the linear orderings with
	$n$-element domains are isomorphic, we denote their order type also by $n$. The order type of the nonnegative
	integers $0<1<\ldots$ is denoted $\omega$ while the order types of the integers and the rationals (equipped
	with their standard ordering relations) are denoted $\zeta$ and $\eta$ respectively. 
	The order type of a linear ordering $(A,<)$ is denoted $o(A,<)$. Since embeddability of
	linear orderings is preserved under isomorphism, this notion can be lifted to order types and so we write
	$o_1\leq o_2$ if the linear orderings of order type $o_1$ can be embedded into the linear orderings of order type
	$o_2$ and $o_1<o_2$ if $o_1\leq o_2$ but not vice versa. Note that the relation $<$ is not necessarily a linear
	ordering on a set of order types, e.g. the intervals $(0,1)$ and $[0,1]$ of real numbers can be embedded into
	each other and they have distinct order types.
	
	When $I=(I,<)$ is a linear ordering and for each $i\in I$, $A_i=(A_i,<_i)$ is a linear ordering, then the
	\emph{(generalized) sum} of the $A_i$s (with respect to $I$) is the linear ordering $\mathop\sum\limits_{i\in I}A_i$
	with domain $\mathop\bigcup\limits_{i\in I}A_i\times\{i\}$ and ordering relation $(x,i)<(y,j)$ if and only if
	$i<j$ or ($i=j$ and $x<_iy$), this ordering called the \emph{anti-lexicographic ordering} of the domain. As a special case, $(A_1,<_1)+(A_2,<_2)$ denotes $\mathop\sum\limits_{i\in\{1,2\}}A_i$.
	When in such a generalized sum all the $A_i$s are the same linear ordering $A$, we write $A\times I$ for
	$\mathop\sum\limits_{i\in I}A$.

	When $I\simeq I'$ and for each $i\in I$, $A_i\simeq A'_i$, then $\mathop\sum\limits_{i\in I}A_i~\simeq~\mathop\sum\limits_{i\in I'}A'_i$ so the sum and product operations extend naturally to
	order types, e.g. $\omega\times 2=\omega+\omega$ is the order type of the linear ordering we get by
	placing two copies of the natural numbers next to each other, while $\omega\times\omega$ is the order type
	of the set consisting of pairs of natural numbers, equipped with the anti-lexicographic order.
	
	A linear ordering $(A,<)$ is a \emph{well-ordering} if it does not contain an infinite descending chain
	$\ldots<x_3<x_2<x_1$, is \emph{quasi-dense} if the rationals, equipped with their standard ordering, can
	be embedded into $A$, and is \emph{scattered} if it is not quasi-dense. A \emph{dense} ordering is a
	linear ordering $(A,<)$ having at least two elements such that whenever $x<y$, then there exists some $z$
	with $x<z<y$. These notions are preserved under isomorphism, so they can be lifted naturally to order types,
	e.g. the finite order types and $\omega$ are well-ordered, $\zeta$ is scattered but not well-ordered,
	and $\eta$ is dense. Each dense ordering is quasi-dense but the converse does not hold, e.g. $2\times\eta$ (the
	ordering we get from the rationals by replacing each rational by two elements) is quasi-dense but not dense.
	(Note that on the other hand, $\eta\times 2=\eta$ holds).
	
	When $(A,<)$ is a linear ordering with order type $o$, then the order type of $(A,<')$ where $x<'y$ if and only
	if $y<x$, is denoted by $-o$, e.g. $-\omega$ is the order type of the negative integers.
	Note that we use here a minus sign instead of the more common notation $o^*$, the reason being to avoid
	confusion with the Kleene star operation.
	
	The order types of well-orderings are called \emph{ordinals}, e.g. $0,1,42,\omega,\omega+3,\omega\times\omega+\omega$
	are ordinals. Since any well-ordered sum or well-orderings is also well-ordered, finite products and sums of 
	ordinals are ordinals as well. Any set of ordinals is well-ordered by the relation $<$ (embeddability in one
	direction), e.g. $0<1<42<\omega<\omega+2<\omega\times\omega$. Moreover, for any set $X$ of ordinals, their
	supremum (with respect to this relation $<$) $\bigvee X$ exists and is an ordinal as well, e.g.
	$\bigvee\{0,1,\ldots\}=\omega$. Each ordinal $\alpha$ is either a \emph{successor ordinal}, in which case
	$\alpha=\beta+1$ for some ordinal $\beta$, or is a \emph{limit ordinal}, in which case $\alpha=\mathop\bigvee\limits_{\beta<\alpha}\beta$. For example, $42$, $\omega+2$ and $\omega\times\omega+3$
	are successor ordinals while $0$, $\omega$ and $\omega\times\omega+\omega$ are limit ordinals. 
	
	As any set of ordinals is well-ordered by $<$, one can use transfinite induction on them, by showing that
	if a property $P$ holds	for some ordinal $\alpha$ then it also holds for $\alpha+1$, and whenever $\alpha$
	is a limit ordinal and $P$ holds for each $\beta<\alpha$, then $P$ holds for $\alpha$ as well, that proves
	$P$ holds for all the ordinals. Frequently, the case when $\alpha=0$ is treated separately.
	
	On ordinals, not only sums and products but exponentiation is also defined: when $\alpha$ and $\beta$
	are ordinals, then the ordinal $\alpha^\beta$ is defined as
	\begin{itemize}
		\item $1$ if $\beta=0$,
		\item $\alpha^\gamma\times\alpha$ if $\beta=\gamma+1$,
		\item $\mathop\bigvee\limits_{\gamma<\beta}\alpha^\gamma$ if $\beta$ is a nonzero limit ordinal.
	\end{itemize}
  The exponentiation notation is ``right-associative'', i.e. $\alpha^{\beta^\gamma}$ stands for $\alpha^{(\beta^\gamma)}$. We omit the parentheses that are redundant applying that $+$ and $\times$ are associative,
  and using the convention that exponentiation takes precedence over product, which in turn takes precedence over sum.
	
 Hausdorff associated an ordinal rank to each scattered ordering (see e.g.~\cite{rosenstein}), but we use a slightly modified variant (not affecting the main result as this variant differs
	from the original one by at most one) introduced in~\cite{10.1007/978-3-642-29344-3_25} as follows. For each ordinal
	$\alpha$ we define a class $H_\alpha$ of linear orderings:
	\begin{itemize}
		\item $H_0$ contains all the finite linear orderings;
		\item $H_\alpha$ for $\alpha>0$ is the least class of linear orderings closed under finite sum and isomorphism
		which contains all the sums of the form $\mathop\sum\limits_{i\in\zeta}A_i$, where for each integer $i$, the
		linear ordering $A_i$ belongs to $H_{\beta_i}$ for some ordinal $\beta_i<\alpha$.
	\end{itemize}
	By Hausdorff's theorem, a countable linear ordering $A$ is scattered if and only if some class $H_\alpha$ contains it:
	the least such $\alpha$ is called the \emph{rank} of the ordering (or of the order type as the value factors through isomorphism) and is denoted $\mathrm{rank}(A)$ (or $\mathrm{rank}(o)$ for the order type $o=o(A)$.)
	
	We note here that the original definition of Hausdorff includes only the empty ordering and the singletons into $H_0$
	and does not require the classes $H_\alpha$ to be closed under finite sum. Since a finite sum of orderings can always be
	written as a zeta-sum of the same orderings and infinitely many zeros, and a zeta-sum of finite linear orderings is also a
	zeta-sum of empty and singleton orderings, this slight change can introduce only a difference of one between the rank,
	e.g. $\omega+\omega$ has rank one in our rank notion but has rank two in the original one. Since $\alpha<o$ for a limit ordinal
	$o$ and an ordinal $\alpha$ if and only if $\alpha+1<o$, and $o=\omega^2$ is a limit ordinal, the main theorem holds for the
	original notion of rank as well.
	
	The reader is encouraged to verify that for any (countable) ordinal $\alpha$,
	the rank of $\omega^\alpha$ is $\alpha$.
	\subsection{Formal languages}
	An \emph{alphabet} is a finite nonempty set of symbols, which are also called \emph{letters} of the alphabet.
    We assume each alphabet comes with a fixed total ordering on its letters.
    For a nonempty set $\Sigma$, $\Sigma^*$ denotes the free monoid freely generated by $\Sigma$,
    that is, the set of (finite) words $a_1\ldots a_n$, $n\geq 0$, $a_i\in\Sigma$ over $\Sigma$. The length
    of a word $w=a_1\ldots a_n$ is $|w|=n$. For $n=0$, we get the empty word which is denoted by $\varepsilon$.
    (We assume the symbol $\varepsilon$ itself is not an element of any alphabet). In this monoid, the product
    operation is $a_1\ldots a_n\cdot b_1\ldots b_k=a_1\ldots a_nb_1\ldots b_k$ and the symbol $\cdot$ is often
    discarded when it does not ruin readability.
    When $u\in\Sigma^*$ is a word and $a\in\Sigma$ is a letter, then $|u|_a$ denotes the number of occurrences
    of $a$ in $u$, formally $|\varepsilon|_a=0$, $|ua|_a=|u|_a+1$ and $|ub|_a=|u|_a$ for each $b\neq a$.
    A \emph{language} over the alphabet $\Sigma$ is any subset
    $L$ of $\Sigma^*$. Product of two languages $K$ and $L$ is defined as $K\cdot L=\{uv:u\in K,v\in L\}$,
    again by omitting the symbol $\cdot$ when there is no danger of confusion.
    
    When $\Sigma$ is a totally ordered set, we use two partial orderings on $\Sigma^*$: the \emph{prefix ordering} $\leq_p$
    ($u\leq_p v$ if and only if $v=uu'$ for some $u'\in\Sigma^*$), with $<_p$ denoting the strict variant of $\leq_p$,
    and the \emph{strict ordering} $<s$ ($u<_sv$ if and only if $u=u_1au_2$ and $v=u_1bu_3$ for some words $u_1,u_2,u_3\in\Sigma^*$ and
    letters $a<b$). Their union is the \emph{lexicographic ordering} $\leq_\ell$ of $\Sigma^*$ which is a total ordering
    and whose strict variant is denoted $<_\ell$. This way, each language $L\subseteq\Sigma^*$ can be viewed as a linear ordering
    set $(L,\leq_\ell)$; let $o(L)$ denote the order type of the language $L$. As an example, for the binary alphabet $\{0,1\}$
    with $0<1$ we have $o(0^*)$ is $\omega$, $o(0^*1)$ is $-\omega$
    as $\ldots<_{\ell} 001<_{\ell} 01<_{\ell} 1$, and $o((00+11)^*01)=\eta$.
    We say that the language $L$ is scattered, well-ordered, etc. if so is the linear ordering $(L,<_\ell)$.

	For a language $L\subseteq \Sigma^*$, we let $\pref(L)$ stand for the set $\{u\in\Sigma^*:~u\leq_pv\hbox{ for some }v\in L\}$ of
	the prefices of the members of $L$. Similarly, let $\mathbf{Suf}(L)$ stand for the set of the suffices of the members of $L$
	(which is formally the reversal of the prefix language of the reversal of $L$, say).
	When $L$ is a language and $n\geq 0$ is an integer, then $L^n$ is the language defined inductively
	as $L^0=\{\varepsilon\}$ and $L^{n+1}=L^nL$, and $L^*$ ($L^+$, resp.) denotes the language
	$\mathop\bigcup\limits_{n\geq 0}L^n$ ($\mathop\bigcup\limits_{n\geq 1}L^n$, resp.) When $L=\{u\}$ is a singleton
    language, we might use $u^*$ and $uK$ for $\{u\}^*$ and $\{u\}K$. The class of \emph{regular} languages
    over some alphabet $\Sigma$ is the least class which contains the empty language $\emptyset$, all the 
    singleton languages $\{a\}$ with $a\in\Sigma$ and which is closed under finite union, product and star.
    When $u\in\Sigma^*$ is a word and $L\subseteq\Sigma^*$ is a language, then $u^{-1}L$ denotes the
    language $\{v\in\Sigma^*:uv\in L\}$. It is known that if $L$ is regular, then so is $u^{-1}L$ for any word $u$.
	
	For each word $u$ there is a shortest prefix $v$ of $u$ so that $u\in v^*$, this word $v$ is called the \emph{primitive root}
	$\mathrm{root}(u)$ of $u$. The word $u$ is called \emph{primitive} if $u=\mathrm{root}(u)$.
	
	An \emph{$\omega$-word} over $\Sigma$ is a sequence $a_1a_2\ldots$ of letters. We let $\Sigma^\omega$
	denote the set of all $\omega$-words. Then, $\Sigma^\omega$ is linearly ordered by the (appropriate modification of the) relation $<_s$ and $\Sigma^*\cup\Sigma^\omega$ is linearly ordered by $<_\ell$. Of course we can
	define the product $u\cdot v\in\Sigma^\omega$ with $u\in\Sigma^*$ and $v\in\Sigma^\omega$ as expected,
	as well as the word $u^\omega=uu\ldots\in\Sigma^\omega$ for each $u\in\Sigma^+$.
	
	\subsection{Transducers and (restricted) one-counter languages}

	Let $D_1\subseteq\{0,1\}^*$ be the language of proper bracketings where
	$0$ plays the role of the opening bracket while $1$ plays
	the closing bracket. That is, a word $u\in\{0,1\}^*$ belongs to $D_1$ if and only if $|u|_0=|u|_1$ and for each
	prefix $v$ of $u$, $|v|_0\geq |v|_1$.
	
	A (nondeterministic) \emph{regular transducer} for the purposes of this paper is a tuple $M=(Q,\Sigma,\Delta,q_0,F,\mu)$
	where $Q$ is the finite set of states, $q_0\in Q$ is the initial state, $F\subseteq Q$ is the set of final states,
	$\Sigma$ is the \emph{output} alphabet, $\Delta\subseteq Q\times \{0,1\}\times Q$ is the transition relation and for each
	$(p,a,q)\in\Delta$, $\mu(p,a,q)$, also denoted $R_{p,a,q}$ is a nonempty regular language over $\Sigma$.
	
	For each word $w\in\{0,1\}^*$ and states $p,q\in Q$ we associate a (regular) language
	$L(M,w,p,q)$ inductively as follows:
	first, let $L(M,\varepsilon,p,q)=\begin{cases}\varepsilon&\hbox{if }p=q\\\emptyset&\hbox{otherwise.}\end{cases}$
	
	Then, for each nonempty word $w=ua$, let
	$L(M,ua,p,q)=\mathop\bigcup\limits_{(r,a,q)\in\Delta}L(M,u,p,r)\cdot R_{r,a,q}$.
	We define $L(M,w)=\mathop\bigcup\limits_{q\in F}L(M,w,q_0,q)$ and
	$L(M)=\mathop\bigcup\limits_{u\in D_1}L(M,u)$. Observe that we only allow the binary alphabet
	as input, moreover, the transducer is by definition only applied to the language $D_1$ of proper
	bracketings -- we make these restrictions to ease notation and to maintain readability of the paper.
	
	A language $L\subseteq\Sigma^*$ is called a \emph{restricted one-counter language} if $L=L(M)$ for some regular transducer $M$.
	As an example, consider the transducer given on Figure~\ref{fig-trans-cban},
	with $q_0$ being its initial and $q_f$ being its only final state. Clearly, only words of the form $w=0^*1^+$ can have a nonempty
	image $L(M,w)$ under $M$, so as $0^*1^+~\cap~D_1=\{0^n1^n:n\geq 1\}$, $L(M)=\mathop\bigcup\limits_{n\geq 1}L(M,0^n1^n)=\mathop\bigcup\limits_{n\geq 1}
	c^n(b^*a)^n$, so this language $L=L(M)$ is a restricted one-counter language. In~\cite{kuske} it has been shown that $o(L)=\omega^\omega$
	and $o(L^k)=\omega^{\omega\times k}$. In particular, for each $k\geq 0$, $L^k$ is a scattered language
	of rank $\omega\times k$. (Note that $L^*$ is not scattered by e.g. Proposition~\ref{prop-iterate-vstar}
	so $L^*$ is \emph{not} an example of a scattered language of rank $\omega^2$, though it's a one-counter
	language, see below.)
	\begin{center}\begin{figure}[h]
		\[\begin{tikzpicture}
		\node[draw,circle] (0) at (0,0) {$q_0$};
		\node[draw,circle] (1) at (3,0) {$q_f$};
		\draw[->] (0) to node[above] {$1~/~b^*a$} (1);
		\draw[->,loop] (0) to node[above] {$0~/~c$} (0);
		\draw[->,loop] (1) to node[above] {$1~/~b^*a$} (0);
		\end{tikzpicture}\]
		\caption{Transducer for $c^n(b^*a)^n$}
		\label{fig-trans-cban}
	\end{figure}
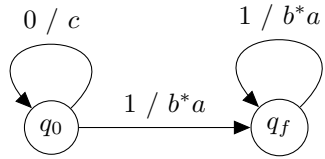
	\end{center}
	A \emph{one-counter language} is usually defined via the means of pushdown automata operating with a single stack symbol.
	The characterization from~\cite{Berstel79transductionsand}, see also
	\cite{DBLP:journals/jcss/Latteux83} suits our purposes better:
	the class of one-counter languages is the least language class
	which contains the restricted one-counter languages and is closed under concatenation, union and Kleene iteration.
\subsection{Linear and semilinear sets}
	Let $\mathbb{N}_0$ stand for the set of nonnegative integers.
	We call a set $X\subseteq\mathbb{N}_0^k$ \emph{periodic} if it has the form
	$X=\{N+M\cdot t:t\geq 0\}$ for some vectors $N,M\in\mathbb{N}_0^k$;
	\emph{linear} if it has the form $X=\{N_0+N_1\cdot t_1+N_2\cdot t_2+\ldots+N_n\cdot t_n:
	t_1,\ldots,t_n\geq 0\}$ for some integer $n\geq 0$ and vectors $N_0,\ldots,N_n\in\mathbb{N}_0^k$;
	\emph{semilinear} if it is a finite union of linear sets
	and \emph{ultimately periodic} if it is a finite union of periodic sets.
	(Observe that a singleton set is also periodic, by choosing the vector $M$
	in the definition to be the null vector, thus finite sets are ultimately periodic.)
	
	It is known~\cite{Matos94periodicsets} that a subset of $\mathbb{N}_0$
	is ultimately periodic if and only if it is semilinear.
	Moreover, by Parikh's theorem we know that the Parikh image $\Psi(L)=\{(|u|_0,|u|_1):u\in L\}$ of any context-free (thus, any regular) language $L\subseteq\{0,1\}^*$ is semilinear (the theorem holds
	for arbitrary alphabets).
	
\section{Some order-theoretic properties of scattered languages and operations}
	In this section we list several statements connecting the rank of scattered languages with language-theoretic
	operations.
	
	The reason why we use the modified rank variant instead of the original one is the following couple
     of handy statements:
\begin{proposition}[\cite{10.1007/978-3-642-29344-3_25}]
	\label{prop-rank-ops}
	Some useful properties of the version of the Hausdorff rank that we use that hold for scattered languages $K$ and $L$:
	\begin{itemize}
		\item $\mathrm{rank}(L)=\mathrm{rank}(\pref(L))$ (in particular, $\mathbf{Pref}(L)$ is also scattered whenever
		$L$ is)
		\item $\mathrm{rank}(K\cup L)=\mathrm{max}\bigl(\mathrm{rank}(K),\mathrm{rank}(L)\bigr)$
		\item $\mathrm{rank}(KL)\leq \mathrm{rank}(L)+\mathrm{rank}(K)$
		\item more generally, if $K$ is scattered of rank $\alpha$ and for each $w\in K$, $L_w$ is a scattered language with rank at most $\beta$,
		then $\mathop\bigcup\limits_{w\in K}wL_w$ is scattered of rank at most $\beta+\alpha$.
	\end{itemize}
\end{proposition}

We make heavy use of the following simple propositions later:
	\begin{proposition}
		\label{prop-iterate-vstar}
		Assume $L\subseteq\Sigma^*$ is a language such that $L^+$ is scattered. Then $L\subseteq v^*$ for some word $v\in\Sigma^*$ (and consequently, so is $L^+$).
	\end{proposition}
	\begin{proof}
		Assume $u,v\in L$ are nonempty words with $\mathrm{root}(u)\neq\mathrm{root}(v)$.
		Then, by Lyndon's theorem (see e.g.~\cite{10.5555/267846}, Theorem 2.2), $uv\neq vu$, say $uv<_svu$ (having the same length, they cannot be in the $<_p$ relation,
		so it's either $uv<_svu$ or the other way around). Then the language $\{uvuv,vuvu\}^*uvvu$ forms a dense subset in $L^+$.
		Thus, if $L^+$ is scattered, then the nonempty members of $L$ share a common primitive root $v$, and hence $L\subseteq v^*$.
	\end{proof}
Languages having a specific form will play crucial role in our proofs:
\begin{definition}
A \emph{prefix chain} is a language $L$ whose words are linearly ordered by the relation $<_p$.
A language $L$ \emph{prefix free} if its words are pairwise incomparable with respect to the relation $<_p$
(and consequently, if and only if it is linearly ordered by the relation $<_s$). 
\end{definition}
Observe that finite prefix chains have finite order types and thus have rank $0$, while infinite prefix
chains have order type $\omega$ and thus have rank $1$. Also, for each infinite prefix chain $C$ there exists
a unique $\omega$-word $w_C$ such that $C\subseteq\mathbf{Pref}(w_C)$. We call $w_C$ the \emph{limit} of $C$.
	\begin{proposition}
		\label{prop-dense-language-has-a-prefixfree-sublanguage}
		If $L\subseteq\Sigma^*$ is a dense language,
		then it has a prefix-free dense subset $K\subseteq L$.
	\end{proposition}
	\begin{proof}
	Let $P\subseteq L$ be the language containing all the words which are members of some infinite prefix chain of
	$L$, that is, $P=\mathop\bigcup\limits_{C\subseteq L\hbox{ is an infinite prefix chain}}C$. Now we have two cases:
	
	{\textbf{Case 1.}} If $P$ is not dense, then there exist two elements $u,v\in P$ such that $u<_\ell v$ but there is no $w\in P$
	with $u<_\ell w<_\ell v$. Then, the sublanguage $L'=\{x\in L ~:~ u <_\ell x <_\ell v \}$ of $L$ is still dense
	and has no member in $P$. In $L'$ there can be elements which are in the prefix relation, but all the
	$<_p$-chains are finite within $L'$ (since if $L'$ contained an infinite $<_p$ chain, its elements would be in $P$). So let $K\subseteq L'$ be the language containing the $<_p$-maximal elements of $L'$. Of course, $K$ is prefix-free.
	We now show that $K$ is dense: let $u_1$, $u_2$ be members of $K$ with $u_1<_\ell u_2$. Since $K$ is prefix-free,
	it has to be the case that $u_1<_su_2$. Now let $u'_2$ be the shortest prefix of $u_2$ with $u_1<_su'_2\leq_p u_2$
	and $u'_2\in L'$. Since $L'$ itself is dense and $u_1,u'_2$ belong to $L'$, there is some word $w\in L'$ with
	$u_1<_\ell w<_\ell u'_2$. With $u_1$ being a $<_p$-maximal element of $L'$, it has to be the case $u_1<_sw$
	and as $w<_pu'_2$ would contradict the minimality of $u'_2$, it also has to be $w<_su'_2$. Hence, $u_1<_sw<_su_2$
	as well and by the assumption of $P$, there has to be a $<_p$-maximal element $w'$ of $L'$ with $w<_p L'$
	(otherwise there would be an infinite prefix chain present in $P$). Hence this $w'$ also belongs to $K$ and so
	$K$ is dense.
	
	{\textbf{Case 2.}} If $P$ is dense, we define a word $x_u\in P$ inductively for each word $u\in \{0,2\}^*\{\varepsilon,1\}$ such that
	$u <_p v$ implies $x_u <_p x_v$ and $u <_s v$ implies $x_u <_s x_v$. This way we embed the dense
	language $\{0,2\}^*\{1\}$ into $P$, proving the statement.
	
	First observe that for each $x\in P$, there has to be an infinite number of $\omega$-words $w$
	such that $x\in\pref(w)$ and $\pref(w)\cap P$ is infinite (that is, there have to be infinitely many
	different prefix chains containing $w$ within $L$), for if there were some $x\in P$ with only a finite number
	of such $\omega$-words, say $\{w_1,\ldots,w_k\}$,
	then by choosing one of them, say $w_1$, there would be a length $N$ such that
	if $u\in\pref(w_1)$ with $|u|\geq N$, then $u\notin\pref(w_i)$ for $i>1$ (as $\omega$-words are linearly ordered
	by $<_s$). Hence, if $u$ and $v$ were
	long enough members of $\pref(w_1)$, then only a finite number of elements of $P$ would fit between them
	(each of them being prefixes of the same $w_1$) and $P$ wouldn't be a dense set.
	
	So, moving back to the construction, for the base step, we choose an arbitrary word from $P$, for $x_\varepsilon$.
	Having defined $x_u\in P$ with $u\in\{0,2\}^*$, we define $x_{u0}$, $x_{u1}$ and $x_{u2}$ as follows.
	Since there are infinitely many infinite prefix chains in $P$ containing $x_u$, we can choose three different
	$\omega$-words, $w_1$, $w_2$ and $w_3$ with $x_u$ being a prefix of each of them and with $w_1<_s w_2<_sw_3$
	and of course with $\mathbf{Pref}(w_i)\cap P$ being infinite (i.e. three $\omega$-words corresponding to
	three different maximal prefix chains within $P$, each containing $x_u$).
	
	Since the three $\omega$-words differ, long enough prefices of $w_i$ are not prefices of the other two words,
	and since each $w_i$ is a limit of an infinite prefix chain, we can choose long enough prefices of each $w_i$
	which are in $P$ and not prefices of the other two $\omega$-words. We define $x_{u0}$, $x_{u1}$ and $x_{u2}$
	to be this prefix of $w_1$, $w_2$ and $w_3$ respectively. Then of course, $x_u<_p x_{u0}, x_{u1}, x_{u2}$
	as well as $x_{u0}<_sx_{u1}<_sx_{u2}$ are satisfied.
	
	Thus, the words of the form $u_{x1}$ form a dense subset of $P$.
	\end{proof}
	\begin{proposition}
		\label{prop-sub-of-scattered-is-scattered}
		If $L\subseteq\Sigma^*$ is a scattered language and $uK\subseteq\pref(L)$ for some word $u\in\Sigma^*$ and
		language $K\subseteq\Sigma^*$, then $K$ is scattered as well.
	\end{proposition}
	\begin{proof}
		Assume $K$ is quasi-dense with $uK\subseteq\mathbf{Pref}(L)$ and let $X\subseteq K$ be a dense subset of $K$.
		Then, $uX$ is still dense and $uX\subseteq\mathbf{Pref}(L)$ which would imply $\mathbf{Pref}(L)$ being
		quasi-dense and thus by Proposition~\ref{prop-rank-ops}, $L$ would have to be quasi-dense as well,
		contradicting the assumptions of the proposition.
	\end{proof}
	\begin{corollary}
		\label{cor-product-members-are-scattered}
		If $L=L_1L_2$ is a nonempty scattered language, then so are $L_1$ and $L_2$.
	\end{corollary}
	\begin{proof}
		Apply Proposition~\ref{prop-sub-of-scattered-is-scattered} with $K=L_1$ and $u=\varepsilon$ for
		$L_1$ and with $K=L_2$ and an arbitrary $u\in L_1$ for $L_2$.
	\end{proof}

\section{The main result}
We are ready to state the main result of the paper. After that, in this we sketch a birds-eye view of its proof,
which is fleshed out in the remaining Sections.
\begin{theorem}
\label{thm-main}
	The rank of any scattered one-counter language is smaller than $\omega^2$.
\end{theorem}
First, observe that it suffices to prove the statement for restricted one-counter languages.
Since a language is one-counter if and only if it can be constructed from restricted one-counter
languages by a finite number of union, product and star applications, we can use induction on the required number
of those applications where the induction steps are
\begin{itemize}
	\item if $K^*$ is a scattered one-counter language for the one-counter language $K$,
	  then by Proposition~\ref{prop-iterate-vstar},
	  $K^*\subseteq v^*$ for some word $v$ and hence has rank at most $1$,
	\item while if $K=L_1L_2$ is a scattered one-counter language for the one-counter languages
	  $L_1$ and $L_2$, then by Corollary~\ref{cor-product-members-are-scattered} both $L_1$ and
	  $L_2$ are scattered. Applying the induction hypothesis we get that the ranks of both $L_1$
	  and $L_2$ are smaller than $\omega^2$, thus applying Proposition~\ref{prop-rank-ops},
	  the rank of $L$ is at most $\mathrm{rank}(L_2)+\mathrm{rank}(L_1)$ that is still smaller than
	  $\omega^2$ if so are the two summands,
	\item and if $K=L_1\cup L_2$ is a scattered one-counter language for the one-counter languages
	  $L_1$ and $L_2$, then again, both of $L_1$ and $L_2$ are scattered as well, thus applying
	  the induction hypothesis and $\mathrm{rank}(K)\leq\mathrm{max}\{\mathrm{rank}(L_1),\mathrm{rank}(L_2)\}$
	  of Proposition~\ref{prop-rank-ops} we get that the rank of $K$ is again smaller than $\omega^2$.
\end{itemize}
Hence it is enough to show that any scattered restricted one-counter language has rank smaller than $\omega^2$.

We prove this in the following way:
\begin{itemize}
	\item We start from a transducer $M$ with $L(M)=L$ being a scattered language.
	\item First we assume that $M$ has the ``feasible cycle property'', stating that whenever
	  there is a cycle in the graph of $M$, then there exists a member $u$ of $D_1$ and a run of $M$
	  over $u$ which visits this cycle.
	\item Then we prove that, by studying the possible cycles in $M$ (which will be categorized to ``$0$-cycles'',
	``positive cycles'' and ``negative cycles'' based on the sign of the difference of the $0$s and $1$s in them)
	that for these transducers, the rank of $L$ has to be smaller than $\omega^2$ (using induction on the height
	of the connected components of the transducer).
	\item Finally, we prove that for each transducer $M$ there exists a transducer $M'$ satisfying the
	feasible cycles property with $L(M)=L(M')$.
\end{itemize}

\section{Transducers with feasible paths only - handling nonnegative cycles}
Let $M=(Q,\Sigma,\Delta,q_0,F,\mu)$ be a transducer. Without loss of generality, we can assume that $q_0$
is a source state with no incoming transitions, $F=\{q_f\}$ is a singleton and $q_f$ is a sink state with
no outgoing transitions. For each transition $(p,a,q)\in\Delta$, let $R_{p,a,q}$ stand for the (regular,
nonempty) output language $\mu(p,a,q)$ and by extension, if $(p,a,q)$ is not a transition in $M$,
then let $R_{p,a,q}$ be $\emptyset$, let $R_{p,\varepsilon,q}=\{\varepsilon\}$ and
for each word $u\in\{0,1\}^*$ and letter $a\in\{0,1\}$, let
$R_{p,ua,q}$ be $\mathop\bigcup\limits_{r\in Q}R_{p,u,r}R_{r,a,q}$. Then, each language $R_{p,u,q}$ is regular
and thus has a finite rank if it is scattered. Also, let $R_{p,q}$ stand for $\mathop\bigcup\limits_{u\in\{0,1\}^*}R_{p,u,q}$.

We also do the similar construction for the input part: let us define the $\leadsto$ relation as follows:
$p\mathop{\leadsto}\limits^\varepsilon q$ if and only if $p=q$, $p\mathop{\leadsto}\limits^aq$ for a letter $a\in\{0,1\}^*$
if and only if $(p,a,q)\in\Delta$ and $p\mathop{\leadsto}\limits^{ua}q$ for $u\in\{0,1\}^*$ and $a\in\{0,1\}$
if and only if there exists a state $r$ with $p\mathop{\leadsto}\limits^ur\mathop{\leadsto}\limits^aq$.
(That is, the ``reachability by an input word'' relation, without taking the output into account.)
Let $p\leadsto q$ denote that $p\mathop{\leadsto}^uq$ holds for some word $u$.
Then $(Q,\leadsto)$ can be seen as a directed transitive graph. Let a \emph{component} of $M$ be a strongly connected
component of this graph (that is, $p$ and $q$ are in the same component of $M$ if $p\leadsto q$ and $q\leadsto p$).

For two states $p,q\in Q$, let us denote by $\mathbf{In}(p,q)=\{u\in\{0,1\}^*:p\mathop{\leadsto}\limits^uq\}$
the language of those input words that can take $p$ to $q$.

A key property of transducers in this section is that of having feasible cycles only:
\begin{definition}
	The transducer $M$ has the \emph{feasible cycles property}, also said \emph{has feasible cycles only}
	if the following conditions hold:
\begin{enumerate}
	\item 
		Whenever $q\in Q$ is a state and $u\in\{0,1\}^*$ is an input word so that $q\mathop{\leadsto}\limits^u q$
		(also called a cycle on $q$),
		then there exist words $v,w\in\{0,1\}^*$ such that $vuw\in D_1$ and $q_0\mathop\leadsto\limits^vq\mathop\leadsto\limits^wq_f$.
	\item 
		Whenever $q\in Q$ is a state and $u\in\{0,1\}^*$ is an input word so that $q\mathop{\leadsto}\limits^uq$
		and $|u|_0>|u|_1$ (also called a positive cycle on $q$), then there exist a word $v$ and for each
		integer $t\geq 0$ a word $w_t$ so that each $vu^tw_t$ is in $D_1$ and
		$q_0\mathop\leadsto\limits^vq\mathop\leadsto\limits^{w_t}q_f$.
	\item
		Whenever $q\in Q$ is a state for which for each integer $N\geq 0$ there exist words $u_N$ and $v_N$
		with $q_0\mathop\leadsto\limits^{u_N}q\mathop\leadsto\limits^{v_N}q_f$, $u_Nv_N\in D_1$
		and $|u_N|_0-|u_N|_1\geq N$, then for each word $u\in\{0,1\}^*$ with $q\mathop\leadsto\limits^{u}q$,
		there exists some integer $N\geq 0$
		and word $v$ with $u_Nuv\in D_1$ and $q\mathop\leadsto\limits^{v}q_f$.
	\end{enumerate}
\end{definition}

We begin studying such transducers by classifying its possible cycles. First we show that if in a component
there exists some \emph{positive cycle}, that is, a state $p$ and a word $u$ with $|u|_0>|u|_1$ and
$p\mathop{\leadsto}\limits^up$, then all the cycles on this state can only generate a language of order type
at most $\omega$ (and thus of rank at most $1$).
\begin{proposition}
	Assume $M$ has feasible cycles only, $L(M)$ is scattered
	and there exists a state $p$ of $M$ and some word $u\in\mathbf{In}(p,p)$
	with $|u|_0>|u|_1$.
	
	Then there exists a primitive word $r_p\in\Sigma^*$ such that $R_{p,p}\subseteq r_p^*$.
\end{proposition}
\label{prop-positive-cycles-imply-rank-1-inner-loops}
\begin{proof}
	Observe that since $q_0$ is a source and $q_f$ is a sink, $p\notin\{q_0,q_f\}$.
	Let us consider the output language $R_{p,u,p}$. Then for each $t\geq 1$ we have
	$R_{p,u,p}^t\subseteq R_{p,u^t,p}$. Moreover, as $M$ has feasible cycles only,
	there exist input words
	$v,w_t\in\{0,1\}$ such that $q_0\mathop{\leadsto}\limits^v p\mathop{\leadsto}\limits^w_t q_f$, and $vu^tw_t\in D_1$, thus in particular, $vu^t\in\mathbf{Pref}(D_1)$. 
	Hence, for an arbitrary output word $v'\in R_{q_0,v,p}$ and integer $t\geq 1$, 
	we get $v'R_{p,u^t,p}R_{p,w_t,q_f}\subseteq L(M)$ and (as $R_{p,u,p}^t\subseteq R_{p,u^t,p}$)
	thus $v'R_{p,u,p}^t\subseteq\mathbf{Pref}(L(M))$. Hence $v'R_{p,u,p}^*\subseteq\mathbf{Pref}(L(M))$, the latter
	being a scattered language by Proposition~\ref{prop-sub-of-scattered-is-scattered}, $R_{p,u,p}^*$ is scattered
	as well, which implies $R_{p,u,p}^*\subseteq r_u^*$ for some primitive output word $r_u\in\Sigma^*$ by
	Proposition~\ref{prop-iterate-vstar}.
	
    Now let $x\in R_{p,p}$, say $x\in R_{p,u_1,p}$ for some word $u_1\in\{0,1\}^*$, be arbitrary.
    Then $u_2=u^{|u_1|+1}u_1$ is also a member of $\mathbf{In}(p,p)$ (as both $u$ and $u_1$ can lead from $p$ to $p$)
    and has more $0$s than $1$s (as $|u|_0>|u|_1$, it is sure that $|u_1|+1$ copies of $u$ has more than $|u_1|$
    more $0$s than $1$s, hence $u^{|u_1|+1}u_1$ still has more $0$s than $1$s).
    Hence,
    $v'(R_{p,u,p}^{|u_1|+1}R_{p,u,p}^*R_{p,u_1,p})^*\subseteq\mathbf{Pref}(L(M))$ which implies
    $R_{p,u,p}^{|u_1|+1}R_{p,u,p}^*R_{p,u_1,p}$ being a subset of $r_p^*$ for some primitive word $r_p\in\Sigma^*$.
    But then, by $R_{p,u,p}^*\subseteq r_u^*$ for the also primitive word $r_u$,
    it has to be the case $r_p=r_u$ and so 
    $R_{p,u_1,p}\subseteq r_u^*$ as well, thus as $x\in R_{p,p}$ was arbitrary we indeed get $R_{p,p}\subseteq r_p^*$
    for the $r_p=r_u$ defined above.
\end{proof}
\begin{corollary}
If $M$ has feasible cycles only and $C$ is a component of $M$ such that for some $p\in C$ there exists an
$u\in\{0,1\}^*$ with $p\mathop{\leadsto}\limits^up$ and $|u|_0>|u|_1$,
then for each state $q\in C$ there exists a primitive word
$r_q\in\Sigma^*$ such that $R_{q,q}\subseteq r_q^*$.
\end{corollary}
\begin{proof}
	Let $p\mathop{\leadsto}\limits^up$, and $q$ be in the same component as $p$, say $q\mathop{\leadsto}\limits^{u_1} p
	\mathop{\leadsto}\limits^{u_2}q$. Then for the word $u'=u_1u^{|u_1u_2|+1}u_2$ we have $q\mathop{\leadsto}\limits^{u'}q$ and $|u'|_0>|u'|_1$: applying Proposition~\ref{prop-positive-cycles-imply-rank-1-inner-loops} on $q$ and $u'$ we get the corollary. 
\end{proof}
Now we turn our attention towards ,,$0$-cycles'', input words $u$ with $|u|_0=|u|_1$ and states $p$ with
$p\mathop{\leadsto}\limits^up$.
\begin{proposition}
\label{prop-zero-cycles-imply-rank-1}
If $M$ has feasible cycles only, $L(M)$ is scattered, $p$ is a state of $M$, and there exists some word $u\in\{0,1\}^+$ with
$|u|_0=|u|_1$ and $p\mathop{\leadsto}\limits^u p$, then there exists a primitive output word $r_p$,
depending only on $p$, such that whenever $|v|_0=|v|_1$ and $p\mathop{\leadsto}\limits^vp$ for a word $v\in\{0,1\}^*$,
then $R_{p,v,p}\subseteq r_p^*$.
\end{proposition}
Note that if there is also a positive cycle on $p$, its $r_p$ from Proposition~\ref{prop-positive-cycles-imply-rank-1-inner-loops} has to be the same as $r_p$ of Proposition~\ref{prop-zero-cycles-imply-rank-1} so using $r_p$ for both Propositions' primitive roots is not
ambiguous.
\begin{proof}
	Analogously to the proof of Proposition~\ref{prop-positive-cycles-imply-rank-1-inner-loops},
	assume $M$, $p$ and $u$ satisfy the conditions of the Proposition. Then, there exist words $u_1$ and $u_2$
	with $q_0\mathop{\leadsto}\limits^{u_1}p\mathop{\leadsto}\limits^up\mathop{\leadsto}\limits^{u_2}q_f$
	and $u_1uu_2\in D_1$, implying $u_1u^tu_2\in D_1$ for each $t\geq 0$.
	Thus if $x\in R_{q_0,u_1,p}$ is arbitrary, we get
	$xR_{p,u,p}^t\subseteq xR_{p,u^t,p}\subseteq\mathbf{Pref}(L(M))$, hence $xR_{p,u,p}^*\subseteq\mathbf{Pref}(L(M))$
	and thus $R_{p,u,p}\subseteq r_p^*$ for some primitive word $r_p\in\Sigma^*$ as $L(M)$ is scattered.
	
	Now if $v\in\{0,1\}^*$ is also some word with $|v|_0=|v|_1$ and $p\mathop{\leadsto}\limits^vp$, then
	there exist some words $v_1,v_2\in\{0,1\}^*$ such that $q_0\mathop{\leadsto}\limits^{v_1}p\mathop{\leadsto}\limits^{v}p\mathop{\leadsto}\limits^{v_2}q_f$.
	Now depending on whether $|v_1|_0-|v_1|_1$ or $|u_1|_0-|u_1|_1$ is greater, for some $w\in\{u,v\}$
	we have $w_1uw_2$ and $w_1vw_2$ both being in $D_1$, hence both $w_1u$ and $w_1v$ are in $\mathbf{Pref}(D_1)$
	and thus by picking an arbitrary $x\in R_{q_0,w_1,p}$ we get that $x(R_{p,u,p}\cup R_{p,v,p})^*\subseteq\mathbf{Pref}(D_1)$, implying $R_{p,u,p}\cup R_{p,v,p}$ being contained in $r_v^*$ for some
	primitive word $r_v\in\Sigma^*$ but as we already know that $R_{p,u,p}\subseteq r_p^*$, it has to be the case
	that $r_v=r_p$, proving the statement of the Proposition.
\end{proof}
There is a third option where cycles in $M$ can only output a language of rank at most one:
\begin{proposition}
	\label{prop-infinite-nx-imply-rank-1}
	Assume $M$ has feasible cycles only, $L(M)$ is scattered and for some state $p$ of $M$ and output
	word $x'\in\Sigma^*$
	it holds that for any integer $N\geq 0$ there exists some words $u_N$ and $v_N$ with
	$q_0\mathop{\leadsto}\limits^{u_N}p$, $p\mathop\leadsto\limits^{v_N}q_f$, 
	$x'\in R_{q_0,u_N,p}$ and $|u_N|_0\geq |u_N|_1+N$.
	
	Then $R_{p,p}\subseteq r_p^*$ for some primitive word $r_p$ (and consequently, for each state $q$
	belonging to the same component of $M$ as $p$, there is a primitive word $r_q$ with $R_{q,q}\subseteq r_q^*$).
\end{proposition}
\begin{proof}
	Let $p\in Q$ and $x'\in\Sigma^*$ satisfy the condition of the Proposition.
	
	Assume $u'_1$ and $v'_1$ are both members of $R_{p,p}$.
	We prove that $x'\{u'_1,v'_1\}^*\subseteq\mathbf{Pref}(L(M))$ which, as $L(M)$ is scattered, implies that $\{u'_1,v'_1\}\subseteq r_p^*$ for some primitive word $r_p\in\Sigma^*$ which implies $u'_1$ and $v'_1$
	have the same primitive root. As $u'_1$ and $v'_1$ were chosen arbitrarily from $R_{p,p}$, all the members
	of $R_{p,p}$ have to have this same primitive root $r_p$, proving the statement.
	
	So let $u_1,v_1\in\{0,1\}^*$ be input words with $u'_1\in R_{p,u_1,p}$ and $v'_1\in R_{p,v_1,p}$
	and let $y'\in\{u'_1,v'_1\}^*$ be arbitrary. Then, $y'\in R_{p,y,p}$ for some $y\in\{u_1,v_1\}^*$.
	Since $M$ has feasible cycles only, there is an integer $N\geq 0$ and word $v$ with $u_Nyv\in D_1$,
	$q_0\mathop\leadsto\limits^{u_N}p\mathop\leadsto\limits^vq_f$.
	Thus, as $x'\in R_{q_0,u_N,p}$, we get $x'y'\in\mathbf{Pref}(L(M))$ and so, as $y'$ was an arbitrary
	member of $\{u'_1,v'_1\}^*$, we get that $x'\{u'_1,v'_1\}^*\subseteq\mathbf{Pref}(L(M))$ indeed holds,
	proving the claim.
\end{proof}
\section{Transducers with feasible cycles only -- the finishing move}

In this section we show that if $M$ is a transducer with feasible cycles only such that $L(M)$ is scattered,
then the rank of $L(M)$ is smaller than $\omega^2$. To this end, for such a transducer $M=(Q,\Sigma,\Delta,q_0,\{q_f\},\mu)$, let us call a transition $\delta=(p,a,q)\in\Delta$ an \emph{intercomponent}
transition if $p$ and $q$ belong to different components of $M$. As $q_0$ is a source and $q_f$ is a sink,
transitions involving these two states are always intercomponent transitions. 
Now let us define for each intercomponent transition $\delta=(p,a,q)$ the language $L(\delta)$ as
\[\mathop\bigcup\left(R_{q_0,u,p}R_{p,a,q}:~u\in\{0,1\}^*:~\exists v\in\{0,1\}^*~uav\in D_1,q_0\mathop\leadsto\limits^up,q\mathop\leadsto\limits^vq_f\right)\]
that is, the language containing all possible output words that are associated with a run in $M$ that starts in $q_0$
and ends in $q$, using the transition $\delta$ as its last step.

We will show that each such $L(\delta)$ is scattered with rank smaller than $\omega^2$. As
$L(M)\subseteq \mathop\bigcup\limits_{\delta=(q,a,q_f)}L(\delta)$ (not necessarily being the same as
the latter language contains not only all the images of the words $u\in D_1$ with $q_0\mathop\leadsto\limits^uq_f$
but also the images of those which are not in $D_1$ but in $\mathbf{Pref}(D_1)$), and each transition arriving
to $q_f$ is an intercomponent one, thus by Proposition~\ref{prop-rank-ops} and that suborderings cannot have
a larger rank we get $\mathrm{rank}(L(M))\leq\max\{\mathrm{rank}(L(\delta)):~\delta=(q,a,q_f)\}<\omega^2$,
proving the main result of the section.

For two intercomponent transitions $\delta=(p,a,q)$ and $\delta'=(p',a',q')$ let us write $\delta<\delta'$
if $q\leadsto p'$. Then this relation $<$ is a strict partial order on the finite set of intercomponent transitions
(as should $q'\leadsto p$ also hold, then the states $p,q,p'$ and $q'$ all belong to the same component of $M$
and thus neither $\delta$ nor $\delta'$ would be intercomponent). We will apply induction with respect to $<$.

For an intercomponent transition $\delta'=(p',a',q')$
with $q'$ being in the same component $C$ as $p$, and a word $u\in\{0,1\}^*$ with $ua'\in\mathbf{Pref}(D_1)$ and $q_0\mathop\leadsto\limits^up'$, let us define the input language $\mathbf{In}(u,\delta',\delta)$ as the set
of those words $v\in\{0,1\}^*$ such that $q'\mathop\leadsto\limits^vp$ and $ua'va\in\mathbf{Pref}(D_1)$,
that is: if a computation path enters $q'$ via $\delta'$ after reading the word $ua'$, then it can read $v$
with being still inside $C$, end in $p$, then leave $C$ via $\delta$ so that the input word $ua'va$ read so far
is still in $\mathbf{Pref}(D_1)$. 

We organize the proof of the induction step into a separate proposition.
\begin{proposition}
	\label{prop-ldelta-inductive}
	Assume $M=(Q,\Sigma,\Delta,q_0,\{q_f\},\mu)$ is a transducer with feasible cycles only and $L(M)$ is scattered.
	Let $\delta=(p,a,q)$ be an intercomponent transition of $M$ and assume for each intercomponent transition
	$\delta'<\delta$, $L(\delta')$ is scattered of rank smaller than $\omega^2$.
	
	Then $L(\delta)$ is also scattered and also has rank smaller than $\omega^2$.
\end{proposition}
\begin{proof}
	The case when $p=q_0$ is clear, with $q_0$ being a source state, $L(\delta)$ is then simply either the
	regular language $R_{q_0,a,q}$ (if $a=0$) or $\emptyset$ (if $a=1$ as no word in $D_1$ can start with
	the letter $1$). As scattered regular languages always have a finite rank, the claim is proved for this
	case.
	
	Now assume $p$ belongs to some component $C\neq \{q_0\}$. Let $\Delta'\subseteq\Delta$ be the
	set of intercomponent transitions entering $C$, that is, of the form $(p',a',q')$ with $p'\notin C$ and
	$q'\in C$. Then,
	\[L(\delta)~=~\mathop\bigcup\limits_{(p',a',q')\in\Delta'}\mathop\bigcup\limits_{u\in\{0,1\}^*:q_0\mathop\leadsto\limits^up',ua'\in\mathbf{Pref}(D_1)}\mathop\bigcup\limits_{v\in\mathbf{In}(u,\delta',\delta)}R_{q_0,u,p'}R_{p',a',q'}
	R_{q',v,p}R_{p,a,q}.\]
	(The reason: any valid computation path over some word within $\mathbf{Pref}(D_1)$
	that leaves $C$ via $\delta$, has to enter $C$ first, via some intercomponent
	transition $\delta'$, after reading in some input word $u$, then taking $\delta'$, after take some route
	within $C$ ending in $p$, reading in some input word $v$ during this phase,
	and then finally taking $\delta$ as well, making sure that the word $ua'va$ read in so far still belongs
	to $\mathbf{Pref}(D_1)$.)
	
	Of course the $R$ languages there are each regular and thus each of them (and their product
	as well) has a finite rank but this does not entail the result as there are infinite unions there. However, the very
	first union is finite as there are only a finite number of transitions, so if we can show that the
	language
	\[\mathop\bigcup\limits_{u\in\{0,1\}^*:q_0\mathop\leadsto\limits^up',ua'\in\mathbf{Pref}(D_1)}\mathop\bigcup\limits_{v\in\mathbf{In}(u,\delta',\delta)}R_{q_0,u,p'}R_{p',a',q'}
R_{q',v,p}R_{p,a,q}\]
	is scattered of rank at most $\omega^2$ for each $(p',a',q')\in\Delta'$, then the statement is proved
	(applying Proposition~\ref{prop-rank-ops}).
	To this end, let us rewrite the above union as follows: for each output word $x\in\Sigma^*$,
	let $U_x\subseteq\{0,1\}^*$ contain those words $u\in\{0,1\}^*$ with $ua'\in\mathbf{Pref}(D_1)$
	for which there exists some $v'\in\{0,1\}^*$ such that $q_0\mathop\leadsto\limits^up'\mathop\leadsto\limits^{a'}q'
	\mathop\leadsto\limits^{v'}q_f$, $ua'v'\in D_1$
	and $x\in R_{q_0,u,p'}$. (Of course $U_x$ might be empty if there is no suitable $u$ at all.)
	Then, we can write the above union as
	\[\mathop\bigcup\limits_{x\in \mathop\bigcup\limits_{u'a'\in\mathbf{Pref}(D_1)}R_{q_0,u',p'}}\left(\mathop\bigcup\limits_{u\in U_x}\mathop\bigcup\limits_{v\in\mathbf{In}(u,\delta',\delta)}xR_{p',a',q'}R_{q',v,p}R_{p,a,q}\right).\]
	The reason why this can help us is the last part of Proposition~\ref{prop-rank-ops}:
	substituting $w=x$, $K=\mathop\bigcup\limits_{u'a'\in\mathbf{Pref}(D_1)}R_{q_0,u',p'}$ and $L_w=\mathop\bigcup\limits_{u\in U_x}\mathop\bigcup\limits_{v\in\mathbf{In}(u,\delta',\delta)}xR_{p',a',q'}R_{q',v,p}R_{p,a,q}$
	we exactly have a language of the form $\mathop\bigcup\limits_{w\in K}wL_w$ here.
	Regarding $R=\mathop\bigcup\limits_{u'a'\in\mathbf{Pref}(D_1)}R_{q_0,u',p'}$, 
	we have $RR_{p',a',q'}=L(\delta')$ which is scattered and has rank smaller than $\omega^2$ by the assumption
	of the Proposition. Thus as $R$ is a subset of $\mathbf{Pref}(L(\delta'))$, it's also scattered and has 
	a rank $\alpha$ which is smaller than $\omega^2$.
	
	So if we manage to show that all the languages of the form $L_w$ above are scattered
	and have some rank smaller than some $\beta$, then by Proposition~\ref{prop-rank-ops} we get that
	the whole union $\mathop\bigcup\limits_{w\in K}L_w$ is scattered of rank at most $\beta+\alpha$.
	Now if $\beta$ is smaller than $\omega^2$ (that is, it has the form $\omega\times k+n$ for some integers
	$k$ and $n$) and so is $\alpha$, then their sum still is smaller than $\omega^2$ and the Proposition is proved.
	
	So let us fix a word $x\in\{0,1\}^*$ belonging to $\mathop\bigcup\limits_{u'a'\in\mathbf{Pref}(D_1)}R_{q_0,u',p'}$
	and consider the language
	\begin{equation}
		\label{eq-distributive-juggling}
		\mathop\bigcup\limits_{u\in U_x}\mathop\bigcup\limits_{v\in\mathbf{In}(u,\delta',\delta)}R_{p',a',q'}R_{q',v,p}R_{p,a,q}~=~R_{p',a',q'}
	\Bigl(\mathop\bigcup\limits_{u\in U_x}\mathop\bigcup\limits_{v\in\mathbf{In}(u,\delta',\delta)}R_{q',v,p}\Bigr)R_{p,a,q}.
	\end{equation}
	That's a product of three languages, with the outermost two being regular, scattered nonempty languages, hence
	having a finite rank. Thus, if the languages of the form
	\begin{equation}
	\label{eq-onlytwounionsleft}
		\mathop\bigcup\limits_{u\in U_x}\mathop\bigcup\limits_{v\in\mathbf{In}(u,\delta',\delta)}R_{q',v,p}
	\end{equation}
	can be shown to have rank $\beta$ smaller than $\omega^2$, then by Proposition~\ref{prop-rank-ops},
	the product of these three languages will have a rank smaller than $n+\beta+k$ for some integers $n$
	and $k$, which is still smaller than $\omega^2$ if so is $\beta$.
	
	Now we rewrite again the union above to a more managable form. Observe that if $u_1$ and $u_2$
	are in $\mathbf{Pref}(D_1)$ with $|u_1|_0-|u_1|_1=|u_2|_0-|u_2|_1$, then for any word $v$,
	$u_1a'va\in\mathbf{Pref}(D_1)$ if and only if $u_2a'va\in\mathbf{Pref}(D_1)$ (the set of possible
	suffixes depends only on the current number of still opened parentheses, assuming the prefix so far
	is valid at all). So let $N_x$ stand for the set $\{|ua'|_0-|ua'|_1:u\in U_x\}$ of nonnegative integers
	and for each integer $n$, let $\mathbf{In}(n,\delta',\delta)\subseteq\{0,1\}^*$ stand for the set
	$\mathop\bigcup\limits_{u\in U_x,|ua'|_0-|ua'|_1=n}\mathbf{In}(u,\delta',\delta)$, that is, the 
	set of those input words $v$ which can lead $M$ from $q'$ to $p$ and then use $\delta$
	while ``closing at most $n$ opening parentheses'',
	i.e. if $|ua'|_0-|ua'|_1=n\geq 0$ for a word $ua'\in\mathbf{Pref}(D_1)$, then $ua'va$ is still in $\mathbf{Pref}(D_1)$.
	We can rewrite Equation~\ref{eq-onlytwounionsleft} as
	\begin{equation}
	\label{eq-againthreeunions}
	\mathop\bigcup\limits_{n\in N_x}\mathop\bigcup\limits_{v\in\mathbf{In}(n,\delta',\delta)}R_{q',v,p}.
	\end{equation}
	Now the part $\mathbf{In}(n,\delta',\delta)$ contains those words $v$ from $\{0,1\}^*$ which can be
	appended after a word in $\mathbf{Pref}(D_1)$ still having $n$ opened parentheses so that the resulting
	word is still in $\mathbf{Pref}(D_1)$, moreover, $v$ can lead from $q'$ to $p$. This is still an infinite
	union from which we aim to create a finite one.
	
	Call a state $r\in C$ \emph{loopable} if $R_{r,r}\subseteq u_r^*$ for some primitive word $u_r$.
	By Propositions~\ref{prop-positive-cycles-imply-rank-1-inner-loops} and 
	\ref{prop-infinite-nx-imply-rank-1} we have that
	\begin{itemize}
		\item if there is some word $u\in\Sigma^*$ with $|u|_0>|u|_1$ and a state $r'\in C$ with
		$r'\mathop\leadsto\limits^ur'$, then all the states of $C$ are loopable;
		\item if $N_x$ is infinite, then all the states of $C$ are loopable.
	\end{itemize}
	Now let us fix a word $v=a_1\ldots a_k\in\mathbf{In}(n,\delta',\delta)$.
	Then $R_{q',v,p}$ is a subset of the union of the languages of the form
	\begin{equation}
	\label{eqn-allowed-product-basic}
		R_{q_1,a_1,q_2}R_{q_2,a_2,q_3}\ldots R_{q_k,a_k,q_{k+1}}
	\end{equation}
	with the union ranging over all the possible sequences $q'=q_1,q_2,\ldots,q_k,q_{k+1}=p$ within $C$.
	For each product of the form \ref{eqn-allowed-product-basic}, there exists at least one, possibly
	more, product of the form
	\begin{equation}
	\label{eqn-allowed-product-advanced}
    R=R_{q'_1,v_1,q'_2}R_{q'_2,v_2,q'_3}\ldots R_{q'_\ell,v_\ell,q'_{\ell+1}}
	\end{equation}
	with $\ell\geq 0$, $q'=q'_1,q'_2,\ldots,q'_\ell,q'_{\ell+1}=p$ being a state sequence within $C$,
	each $v_i\in\{0,1\}^+$ being a word with $v=v_1\ldots v_\ell$ such that whenever $|v_i|>1$,
	then $q'_i=q'_{i+1}$ and either $q'_i$ is a loopable state, or $|v_i|_0=|v_i|_1$ (or both),
	moreover, $R_{q_1,a_1,q_2}\ldots R_{q_k,a_k,q_{k+1}}\subseteq R$,
	with one possible such product being the decomposition~\ref{eqn-allowed-product-basic} itself.
	Let us choose one such product $R_1\ldots R_\ell$ satisfying the above conditions for~\ref{eqn-allowed-product-advanced} which minimizes $\ell$.
	
	Now we bound $\ell$ in terms of $|C|$ and $n$.
	Observe that if for a product of the form~\ref{eqn-allowed-product-advanced}
	there exist $i<j$ with $q'_i=q'_{j+1}$ either being a loopable state, or with
	$v_i\ldots v_j$ having the same number of $0$s and $1$s, then the factor sequence
	$R_{q'_i,v_i,q_{i+1}}\ldots R_{q'_j,v_j,q'_{j+1}}$ can be replaced to its superset $R_{q'_i,v_i\ldots v_j,q'_{j+1}}$,
	lowering the number of factors and still producing a product formed satisfying the condition, so
	in a shortest product of the form \ref{eqn-allowed-product-advanced} no loopable states get repeated
	and there is no $0$-cycle either spanning over more factors.
	
	We do now a case analysis.
	
	{\textbf{Case 1.}} If all the states are loopable, this means $\ell$ is at most $2|C|$ as no $q'_i$ can
	be the same as $q'_j$ for any $i<j$ unless $j=i+1$. Hence the longest possible shortest sequence can have
	$q'_1=q'_2$, then $q'_3=q'_4$, and so on, for $q'_{2|C|-1}$ and $q'_{2|C|}$ to finish the sequence (enumerating
	each state of $C$ in some order, spelling each state twice).
	
	{\textbf{Case 2.}} If not all the states are loopable, then $N_x$ is finite and there are absolutely no
	words $u$ and states $r\in C$ with $r\mathop\leadsto\limits^ur$ and $|u|_0>|u|_1$ (due to Propositions
	~\ref{prop-positive-cycles-imply-rank-1-inner-loops} and~\ref{prop-infinite-nx-imply-rank-1}).
	Then, each factor of the form $R_{r,v,r}$ with $|v|>1$ has $|v|_0=|v|_1$ and all other factors have
	the form $R_{q_i,a_i,q_{i+1}}$ for some $i$. Now assume $\ell>(n+1+|C|)\cdot|C|$. Then there is a state
	$r$ which appears at least $n+2+|C|$ times in the sequence $q'_1,\ldots,q'_{\ell+1}$.
	By minimality of $\ell$, if $q'_i=q'_{j+1}$ for some $i<j$, then for the word $v[i,j]=v_iv_{i+1}\ldots v_j$
	we have $|v[i,j]|_0<|v[i,j]|_1$ (since if it were the other way around, we would have a positive cycle
	and if they were equal, we could collapse this interval into $R_{q'_i,v[i,j],q'_{j+1}}$). Hence
	if $r$ appers at least $n+2+|C|$ times in the sequence, with $q'_i$ being its first appearance and $q'_{j+1}$
	being the last one, then for the word $v[i,j]$ we have $|v[i,j]|_0+n+1+|C|<|v[i,j]|_1$. Since we started
	from a word $v\in\mathbf{In}(n,\delta',\delta)$, it has to be the case that $|v[1,j]|_0+n\geq |v[1,j]|_1$
	as otherwise $v$ could not be in any set $\mathbf{In}(u,\delta',\delta)$ with $|ua'|_0-|ua'|_1=n$ as
	the word $ua'v[1,j]$, which is a prefix of $ua'va\in\mathbf{Pref}(D_1)$ would contain more $1$s than $0$s
	which cannot happen. Hence, from $|v[i,j]|_0+n+1+|C|<|v[i,j]|_1$ and $|v[1,j]|_0+n\geq |v[1,j]|_1$
	we get $|v[1,i-1]|_0> |v[1,i-1]|_1+|C|$, that is, in the prefix $v[1,i-1]$ there have to be much more $0$s
	than $1$s to be able to handle all the $1$s arriving later with still remaining in $\mathbf{Pref}(D_1)$.
	However, whenever $|v_k|_0\neq |v_k|_1$, then (since we are within Case 2) $v_k$ has to be one of the symbols
	$0$ or $1$, so the difference between $|v[1,k]|_0-|v[1,k]|_1$ and $|v[1,k+1]|_0-|v[1,k+1]|_1$ is at most $1$
	for each index $k$ (and starts from $0$ as $v[1,0]$ can be seen as the empty word).
	
	Now this means that if we reach up to $|C|+1$ with this difference, then there are indices $i_1<i_2<\ldots<i_{|C|+1}$
	in the sequence such that for each $k=1,\ldots,|C|+1$, $i_k$ is the first index satisfying $|v[1,i_k]|_0-|v[1,1_{k+1}]|_1=k$. Hence, as this sequence has $|C|+1$ elements, there has to be a state
	which gets repeated but if (say) $q'_{i_k}=q'_{i_m}$ for $k<m$, then for the word $w=v[i_k,i_{m-1}]$ we would get
	$q'_{i_k}\mathop\leadsto\limits^w q'_{i_k}$ with $|w|_0-|w|_1>0$, contradicting to the assumption we are
	in Case 2.
	
	Hence, in this case $\ell$ is at most $(n+1+|C|)\cdot|C|$. Note that this quantity does not depend on $v$
	anymore, only on $n$ (which depends on the word $x$) and on $|C|$ (which depends on $\delta$).
	
	Thus we now know that each product of the form~\ref{eqn-allowed-product-basic} is a subset of a product
	of the form
	$R_1\ldots R_\ell$ with each $\ell$ being at most $(n+1+|C|)\cdot |C|$ and each $R_i$ being
	either an $R_{q'_i,a'_i,q'_{i+1}}$ for some states $q'_i$ and $q'_{i+1}$ of $C$ and letter $a'_i\in\{0,1\}$,
	or an $R_{q'_i,v_i,q'_i}$ with either $q'_i$ being a loopable state, or $|v_i|_0=|v_i|_1$, in both cases
	being a scattered language of rank at most $1$ (as in both cases, these languages are subsets of 
	$w^*$ for some appropriate primitive word $w$ due to Propositions~\ref{prop-positive-cycles-imply-rank-1-inner-loops} and \ref{prop-zero-cycles-imply-rank-1}).
	The other languages $R_{q'_i,a'_i,q'_{i+1}}$ are all regular languages, defined by the transducer and there
	is only a finite number of them. Hence, there is an absolute constant integer $N\geq 0$ depending only on
	$M$ such that each language
	$R_{q'_i,v_i,q'_i}$ is scattered of at most $N$, thus by Proposition~\ref{prop-rank-ops} the rank of
	each such product $R_1\ldots R_\ell$ is bounded by either $2\cdot|C|\cdot N$ (in Case 1) or by 
	$(n+1+|C|)\cdot |C|\cdot N$ (in Case 2). As there are only $|\Delta|$ ``elementary'' languages and at most $|C|$
	languages of either the form $R_{q'_i,q'_i}$ (when $q'_i$ is loopable) or $\mathop\bigcup\limits_{|u|_0=|v|_0}R_{q'_i,u,q'_i}$ (when $v_i$ is a $0$-cycle), if we let
	$R'$ denote the union of all the languages of the form $R_{q'_i,a'_i,q'_{i+1}}$, $R_{q'_i,q'_i}$ for $q'_i$ loopable
	and $\mathop\bigcup\limits_{|u|_0=|v|_0}R_{q'_i,u,q'_i}$, we get that $R'$ is a finite union of languages
	of finite rank, and also for each $n\in N_x$, $\mathop\bigcup\limits_{v\in\mathbf{In}(n,\delta',\delta)}R_{q',v,p}$ is a subset of $R'^\ell$ for the
	power $\ell$ computable from $M$ and $n$, hence these languages are products of scattered languages of
	finite rank, thus they are also scattered of finite rank as well.
	
	When $N_x$ is finite, this makes the language of Equation~\ref{eq-againthreeunions} to be a subset of a
	finite union of scattered languages, each having a finite rank and we are done proving that the languages
	of the form Equation~\ref{eq-againthreeunions} are always scattered and have a finite rank.
	
	Now, $N_x$ in Equation~\ref{eq-againthreeunions} might be infinite but in that case all the states are
	loopable and $\ell\leq 2\cdot |C|$ does not depend on $n$. Hence in that case the whole union itself
	is a subset of the language $R'^{2\cdot |C|}$ which is again a scattered language of finite rank as so is $R'$.
	
	In summary, we proved that the languages in Equation~\ref{eq-againthreeunions} are scattered and have a finite rank, which in turn implies the languages of Equation~\ref{eq-onlytwounionsleft} are also scattered and have a finite
	rank, hence their rank (which was called $\beta$ just above Equation~\ref{eq-distributive-juggling})
	is indeed smaller than $\omega^2$ (smaller than $\omega$ actually), finishing the proof of the Proposition.	
\end{proof}
\begin{corollary}
	\label{cor-feasible-are-small}
	Suppose $M$ is a transducer having feasible cycles only and $L(M)$ is scattered.
	Then the rank of $L(M)$ is smaller than $\omega^2$.
\end{corollary}
\begin{proof}
	Applying Proposition~\ref{prop-ldelta-inductive} as the inductive step, we get that for each intercomponent
	transition $\delta$, the language $L(\delta)$ is scattered and has a rank smaller than $\omega^2$.
	As $L(M)$ itself is a subset of the finite union $\mathop\bigcup\limits_{\delta=(q,a,q_f)\in\Delta}L(\delta)$,
	the claim is proved applying Proposition~\ref{prop-rank-ops}.
\end{proof}

\section{Making all cycles feasible}

	In this section we show that for every transducer $M$ there is a transducer $M'$ having feasible
	cycles only with $L(M)=L(M')$, assuming $L(M)\neq\emptyset$.
	Together with Corollary~\ref{cor-feasible-are-small} it implies that
	scattered restricted one-counter languages have rank smaller than $\omega^2$, which in turn implies
	Theorem~\ref{thm-main}.

	Let us define the (net) \emph{opening depth} of a word $w\in\{0,1\}^*$ as
	$\mathrm{open}(w)=|w|_0-|w|_1$. Clearly, a word $w$ belongs to $\mathbf{Pref}(D_1)$
	if and only if $\mathrm{open}(w')\geq 0$ for each prefix $w'$ of $w$, and
	to $D_1$ if additionally, $\mathrm{open}(w)=0$. As an extension, we define
	$\mathrm{open}':\mathbb{N}_0^2\to\mathbb{N}_0$ as $(n,m)\mapsto n-m$.
	Then clearly, $\mathrm{open}(w)=\mathrm{open}'(\Psi(w))$ for each word $w\in\{0,1\}^*$
	(recall that $\Psi(w)=(|w|_0,|w|_1)$ is the Parikh image of $w$)
	and the image under $\mathrm{open}'$ of a linear set $\{(n_0,m_0)+(n_1,m_1)\cdot t_1+\ldots+(n_k,m_k)\cdot t_k:t_1,\ldots,t_k\geq 0\}\subseteq\mathbb{N}_0^2$ is the linear (thus ultimately periodic)
	set  $\left\{(n_0-m_0)+\mathop\sum\limits_{i=1}^k(n_i-m_i)\cdot t_i:t_1\ldots,t_k\geq 0\right\}\subseteq\mathbb{N}_0$. Hence, $\mathrm{open}(L)$ is an ultimately periodic
	set for any context-free language $L\subseteq\{0,1\}^*$, in particular, for $D_1$,
	$\mathbf{Pref}(D_1)$, $\mathbf{Suf}(D_1)$, their intersections with regular languages,
	and finite unions and products of such languages.
	
	Similarly, let us define the \emph{closing depth} of a word $w\in\{0,1\}^*$ as
	$\mathrm{close}(w)=|w|_1-|w|_0$. Then, a word $w$ belongs to $\mathbf{Suf}(D_1)$
	if and only if $\mathrm{close}(w')\geq 0$ for each suffix $w'$ of $w$, and
	belongs to $D_1$ if and only if additionally $\mathrm{close}(w)=0$.
	Again, we define $\mathrm{close'}(n,m)=m-n$. We get also that for any context-free
	language $L\subseteq\{0,1\}^*$, $\mathrm{close}(L)\subseteq\mathbb{N}_0$ is ultimately periodic.
	
	Given a transducer $M=(Q,\Sigma,\Delta,q_0,F,\mu)$,
	we associate to each state $q\in Q$ the following sets $N_{-}(q), N_{+}(q)$ and $N(q)\subseteq\mathbb{N}_0$
	of integers:
	\begin{itemize}
		\item $n\in N_{-}(q)$ if and only if there exists some $u\in\mathbf{Pref}(D_1)$ with $\mathrm{open}(u)=n$
		and $q_0\mathop\leadsto\limits^uq$
		\item $n\in N_{+}(q)$ if and only if there exists some $v\in\mathbf{Suf}(D_1)$ with $\mathrm{close}(v)=n$
		and $q\mathop\leadsto\limits^vq_f$ for some $q_f\in F$
		\item and $N(q)=N_{-}(q)\cap N_{+}(q)$.
	\end{itemize}
	Then, e.g. $n\in N(q)$ if and only if there exists at least some successful computation path in $M$
	reading in some word $uv\in D_1$ for which after reading $u$ in, the path is in $q$ and there are
	exactly $n$ open parentheses at that instant. 
	
	\begin{proposition}
	\label{prop-nq-ultimately-periodic}
		For each state $q$ of a transducer $M$, the sets $N(q)$, $N_{-}(q)$ and $N_{+}(q)$ are ultimately periodic.
	\end{proposition}
	\begin{proof}
		As $N_{-}(q)=\mathrm{open}(\{u\in\mathbf{Pref}(D_1):q\in q_0u\})$ and this
		language is the intersection of the context-free language $\mathbf{Pref}(D_1)$
		and the regular language $\{u\in\{0,1\}^*:q_0\mathop\leadsto\limits^uq\}$, thus is context-free
		as well, we have that $N_{-}(q)$ is ultimately periodic.
		
		Similarly, $N_{+}(q)$ is ultimately periodic as well. As the intersection of
		finitely many ultimately periodic sets is ultimately periodic~\cite{Matos94periodicsets},
		so is $N(q)$.
	\end{proof}
	For an example for a transducer (without the output function as that does not play a role in the
	sets $N(q)$) and the sets $N(q)$ see Figure~\ref{fig-nq}. The reader is encouraged to verify some of
	these sets, e.g. for $N_+(q_1)$ we have that the words accepted from $q_1$ are the members of the 
	language $(000+01)^*0(1(11)^*+11)~\cap~\mathbf{Suf}(D_1)$ on which if we apply the $\mathrm{close}$ function we get the 
	nonnegative numbers belonging to the set $\{-3t_1-1+1+2t_2:t_1,t_2\geq 0\}~\cup~\{-3t_1-1+2:t_1\geq 0\}$,
	that is, $\{2t_2-3t_1:t_1,t_2\geq 0,2t_2\geq 3t_1\}~\cup~\{1\}$ which in turn is simply $\mathbb{N}_0$, or
	$\{t:t\geq 0\}$ as each nonnegative integer $k$ can be written as either $k=2\cdot t_2-3\cdot 0$ if $k$ is even
	and as $k=2t_1-3\cdot 1$ if $k$ is odd.
		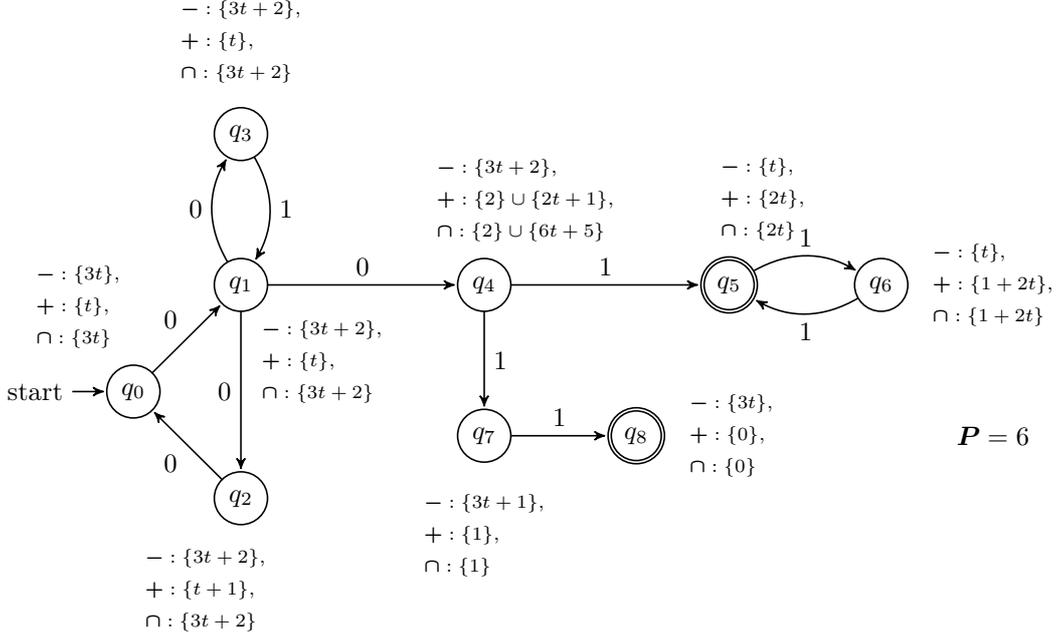
\begin{figure}
		\begin{center}
			\begin{tikzpicture}[-,>=stealth',shorten >=1pt,auto,node distance=2cm, semithick, every text node part/.style={align=left}]
			\tikzstyle{every state}=[draw=black,text=black, inner sep = 0cm, outer sep = 0cm, minimum size = 0.7 cm]
			\node[initial,state] (q0)               {$q_0$};
			\node[state]         (q1) [above right of=q0] {$q_1$};
			\node[state]         (q2) [below right of=q0] {$q_2$};
			\node[state]         (q3) [above of=q1] {$q_3$};
			\node[state]         (q4) [right = 2.5cm of q1]  {$q_4$};
			\node[state,accepting]         (q5) [right = 2.5cm of q4] {$q_5$};
			\node[state]   (q6) [right of=q5]       {$q_6$};
			\node[state]         (q7) [below of=q4]       {$q_7$};
			\node[state,accepting]   (q8) [right of=q7]       {$q_8$};
			
			\path[->]
			(q0) edge node {$0$} (q1)
			(q1) edge node[left] {$0$} (q2)
			(q2) edge node {$0$} (q0)
			(q1) edge[bend left=30] node {$0$} (q3)
			(q3) edge[bend left=30] node {$1$} (q1)
			(q1) edge node {$0$} (q4)
			(q4) edge node {$1$} (q5)
			(q4) edge node {$1$} (q7)
			(q7) edge node {$1$} (q8)
			(q5) edge[bend left=30] node {$1$} (q6)
			(q6) edge[bend left=30] node {$1$} (q5)
			;
			
			\node[rectangle, draw=none] (sq0) [above left = 0.2cm and -0.2cm of q0] { \scriptsize$\boldsymbol{-}: \{3t\}$,\\ \scriptsize  $\boldsymbol{+}:\{t\},$\\ \scriptsize $\boldsymbol{\cap}: \{3t\}$ };
			\node[rectangle, draw=none] (sq1) [below right =0.1cm and -0.1cm of q1] { \scriptsize$\boldsymbol{-}: \{3t+2\}$,\\ \scriptsize  $\boldsymbol{+}: \{t\},$\\ \scriptsize $\boldsymbol{\cap}: \{3t+2\}$ };
			\node[rectangle, draw=none] (sq2) [below left=0.3cm and -0.7cm of q2] { \scriptsize$\boldsymbol{-}: \{3t+2\}$,\\ \scriptsize  $\boldsymbol{+}:\{t+1\},$\\ \scriptsize $\boldsymbol{\cap}: \{3t+2\}$ };
			\node[rectangle, draw=none] (sq3) [above = 0.2cm of q3] { \scriptsize$\boldsymbol{-}: \{3t+2\}$,\\ \scriptsize  $\boldsymbol{+}:\{t\},$\\ \scriptsize $\boldsymbol{\cap}: \{3t+2\}$ };
			\node[rectangle, draw=none] (sq4) [above right = 0.2cm and -1.0cm of q4] { \scriptsize $\boldsymbol{-}: \{3t+2\}$,\\ \scriptsize  $\boldsymbol{+}:\{2\}\cup\{2t+1\},$\\ \scriptsize $\boldsymbol{\cap}: \{2\}\cup\{6t+5\}$ };
			\node[rectangle, draw=none] (sq5) [above right = 0.2cm and -0.5cm of q5] { \scriptsize$\boldsymbol{-}: \{t\}$,\\ \scriptsize  $\boldsymbol{+}:\{2t\},$\\ \scriptsize $\boldsymbol{\cap}: \{2t\}$ };
			\node[rectangle, draw=none] (sq6) [right = 0.2cm of q6] { \scriptsize$\boldsymbol{-}: \{t\}$,\\ \scriptsize  $\boldsymbol{+}:\{1+2t\},$\\ \scriptsize $\boldsymbol{\cap}: \{1+2t\}$ };
			\node[rectangle, draw=none] (sq7) [below = 0.3cm of q7] { \scriptsize$\boldsymbol{-}: \{3t+1\}$,\\ \scriptsize  $\boldsymbol{+}:\{1\},$\\ \scriptsize $\boldsymbol{\cap}: \{1\}$ };
			\node[rectangle, draw=none] (sq8) [right =0.2cm of q8] { \scriptsize$\boldsymbol{-}: \{3t\}$,\\ \scriptsize  $\boldsymbol{+}:\{0\},$\\ \scriptsize $\boldsymbol{\cap}: \{0\}$ };
			
			 \node (p) [below of=sq6] {$\boldsymbol{P}=6$};
			\end{tikzpicture}\end{center}
			\caption{The sets $N_-(q)$, $N_+(q)$ and $N(q)$, denoted by $-$, $+$ and $\cap$ respectively.}
			\label{fig-nq}
	\end{figure}

	\begin{proposition}
	\label{prop-period-and-tau-exist}
		For any transducer $M$, there exists some integer $P>1$, called a \emph{period} of $M$
		and for each state $q$ of $M$,
		some subset $\tau(q)$ of $\{0,\ldots,2P-1\}$, called the \emph{type} of $q$
		such that
		\[N(q)=\bigl(\tau(q)\cap\{0,\ldots,P-1\}\bigr)~\cup~\{n\in\mathbb{N}:~n\geq P,n\equiv r~\mathrm{mod}~P\hbox{ for some }r\geq P,r\in\tau(q)\}.\]
	\end{proposition}
	\begin{proof}
		By Proposition~\ref{prop-nq-ultimately-periodic}, each set $N(q)$ is ultimately periodic,
		that is, a finite union of sets of the form $\{r+p\cdot t:t\geq 0\}$ for some constants
		$r,p\geq 0$ (called the remainder and the period -- the case $p=0$ defines a singleton set).
		Let $P$ be the least integer which is a multiple of each nonzero period and larger than
		all the remainders and is also at least two.
		
		We claim that $X(q)=\{n:0\leq n\leq 2P-1\}\cap N(q)$ is a good choice for the type of $q$.
		To this end, let $\widehat{X}(q)$ stand for the (ultimately periodic) set 
		\[\bigl(X(q)\cap\{0,\ldots,P-1\}\bigr)~\cup~\mathop\bigcup\limits_{r\in X(q),r\geq P}\{n\geq P:n\equiv r~\mathrm{mod}~P\}.\]
		So we have to show that $N(q)=\widehat{X}(q)$.
		
		First, observe that $\widehat{X}(q)\cap\{0,\ldots,P-1\}~=~N(q)\cap\{0,\ldots,P-1\}$ by the definition of
		$X(q)$ so we have to show that for any integer $n\geq P$, $n\in\widehat{X}(q)$ if and only if $n\in N(q)$.
		Let us write $N(q)=\mathop\bigcup\limits_{i\in[k]}\{r_i+p_i\cdot t:t\geq 0\}$
		
		And indeed, for $n\geq P$ (and thus $n\geq r_i,p_i$ for each $i\in[k]$) we have
		\begin{align*}
		n\in\widehat{X}(q) &\Leftrightarrow n\equiv r~\mathrm{mod}~P\hbox{ for some }r\in X(q),r\geq P\\
		&\Leftrightarrow n\equiv r~\mathrm{mod}~P\hbox{ for some }r\in N(q),P\leq r<2P\\
		&\Leftrightarrow n\equiv r_i+p_i\cdot t~\mathrm{mod}~P\hbox{ for some }i\in[k], 0\leq t\\
		&\Leftrightarrow n\equiv r_i+p_i\cdot t~\mathrm{mod}~P\hbox{ for some }i\in[k], 0\leq t<P/p_i\\
		&\Leftrightarrow n\equiv r_i~\mathrm{mod}~p_i,n\geq r_i\hbox{ for some }i\in[k]\\
		&\Leftrightarrow n\in N(q).
		\end{align*}
	\end{proof}

	{\textbf{Now we create a transucer $M'$ from $M$ by creating copies of each state.}}
	The states of $M'$ will be triples of the form $(q,n,\sigma)$ with $q\in Q$, $n\in\tau(q)$ and $\sigma\in\{\equiv,\uparrow,\downarrow\}$.
	
	Let $P$ be a period of $M$. From the state $q$ of $M$, we will create states $(q,n,\equiv)$ for each $P\leq n\in\tau(q)$
	and two states, $(q,n,\uparrow)$ and $(q,n,\downarrow)$ for each $n\in\tau(q)$ with $n<P$.
	Observe that since $q_0\mathop\leadsto\limits^wq_f$ for some $w\in D_1$ (otherwise $L(M)$ is empty)
	and $q_f\in F$,
	ˇwe have $0\in\tau(q_0)$.
	In $M'$, let $(q_0,0,\uparrow)$ be the initial state.
	Also, if $q_f\in F$, then we can assume that there exists some word $w\in D_1$ with
	$q_0\mathop\leadsto\limits^wq_f$ (otherwise we can remove $q_f$ from $F$, the resulting transducer will
	be equivalent with $M$), and so $0\in N(q_f)$ as well. So let $\{(q_f,0,\downarrow):q_f\in F\}$ be the
	(nonempty) set of accepting states in $M'$.
	
	We define the transitions of $M'$ as follows: let $((p,n,\sigma_1),a,(q,m,\sigma_2))\in\Delta'$
	if and only if $(p,a,q)\in\Delta$ and one of the following conditions holds:
	\begin{enumerate}
		\item[i)] $n+1=m<P$, $\sigma_1=\sigma_2$ and $a=0$
		\item[ii)] $n-1=m$, $m<P$, $\sigma_2\in\{\sigma_1,\downarrow\}$ and $a=1$
		\item[iii)] $n+1\equiv m~\mathrm{mod}~P$, $m\geq P$, $n\geq P-1$, $a=0$, $\sigma_2=\equiv$
		and $\sigma_1\neq\downarrow$
		\item[iv)] $n-1\equiv m~\mathrm{mod}~P$, $n\geq P$, $m\geq P-1$, $a=1$, $\sigma_1=\equiv$
		and $\sigma_2\neq\uparrow$.
	\end{enumerate}
	Moreover, for $((p,n),a,(q,m))\in\Delta'$, let $\mu'((p,n),a,(q,m))=\mu(p,a,q)$.
	Finally, if there is any non-accessible or non-coaccessible state in $M'$, then let us drop it.
	Figure~\ref{fig-mprime} shows a part of the transducer $M'$ constructed from the transducer $M$ of Figure~\ref{fig-nq}
	with some states missing and without the output function, to maintain readability of the transition diagram.
	
	The idea is that when $M'$ reads some input word, then for a while it uses states labeled by $\uparrow$, then
	if for the currently read prefix the opening depth reaches $P$, then from that point it uses states labeled by
	$\equiv$, then, after reading in the longest prefix with opening depth at least $P$ it switches to states
	labeled by $\downarrow$. In the $\uparrow$ and $\downarrow$ states, the exact opening depth is maintained while
	in the $\equiv$ states it's maintained only up to modulo $P$. (During the switch from an $\equiv$ state to a $\downarrow$
	state, nondeterminism is used to guess the end of the longest prefix and this guess is then checked against by the
	$\downarrow$ states.) Finally, if the depth of the word never reaches $P$, then the transducer switches at some point
	from an $\uparrow$-state to a $\downarrow$ state by a transition of type ii). Most of these latter
	transitions are missing intentionally from the diagram of $M'$ of Figure~\ref{fig-mprime}.
		
	\begin{figure}
	\begin{center}
	\begin{tikzpicture}[scale=0.8, transform shape, -,>=stealth',shorten >=1pt,auto,node distance=1cm, semithick, every text node part/.style={align=left}]
	\tikzstyle{every state}=[draw=black,text=black, inner sep = 0cm, outer sep = 0cm, minimum size = 0.6 cm, ellipse]
	\node[initial,state] (q0-0)                 {$q_0,0,\uparrow$};
	\node[state]         (q0-3) [above of=q0-0] {$q_0, 3, \uparrow$};
	\node[state]         (q0-6) [above of=q0-3] {$q_0, 6, \equiv$};
	\node[state]         (q0-9) [above of=q0-6] {$q_0, 9, \equiv$};
	
	\node[state] 		 (q1-1) [below right = 2cm and 2cm of q0-0]  {$q_1,1,\uparrow$};
	\node[state]         (q1-4) [above of=q1-1] {$q_1, 4, \uparrow$};
	\node[state]         (q1-7) [above of=q1-4] {$q_1, 7, \equiv$};
	\node[state]         (q1-10) [above of=q1-7] {$q_1, 10, \equiv$};
	
	\node[state] 		 (q2-2) [below left = 5cm and 1cm of q0-0]  {$q_2,2,\uparrow$};
	\node[state]         (q2-5) [above of=q2-2] {$q_2, 5, \uparrow$};
	\node[state]         (q2-8) [above of=q2-5] {$q_2, 8, \equiv$};
	\node[state]         (q2-11) [above of=q2-8] {$q_2, 11, \equiv$};
	
	\node[state] 		 (q3-2) [right = 2cm of q1-1]  {$q_3, 2,\uparrow$};
	\node[state]         (q3-5) [above of=q3-2] {$q_3, 5, \uparrow$};
	\node[state]         (q3-8) [above of=q3-5] {$q_3, 8, \equiv$};
	\node[state]         (q3-11) [above of=q3-8] {$q_3, 11, \equiv$};
	
	\node[state] 		 (q4-2) [below right = 4cm and -1cm of q3-2]  {$q_4, 2,\uparrow$};
	\node[state]         (q4-5) [above of=q4-2] {$q_4, 5, \uparrow$};
	\node[state]         (q4-11) [above of=q4-5] {$q_4, 11, \equiv$};
	
	\node[state,accepting] 		 (q5-0) [below right =1cm and 2cm of q4-2]  {$q_5, 0,\downarrow$};
	\node[state]         (q5-2) [above of=q5-0] {$q_5, 2, \downarrow$};
	\node[state]         (q5-4) [above of=q5-2] {$q_5, 4, \downarrow$};
	\node[state]         (q5-6) [above of=q5-4] {$q_5, 6, \equiv$};
	\node[state]         (q5-8) [above of=q5-6] {$q_5, 8, \equiv$};
	\node[state]         (q5-10) [above of=q5-8] {$q_5, 10, \equiv$};
	
	\node[state] 		 (q6-1) [right = 2cm of q5-0]  {$q_6, 1,\downarrow$};
	\node[state]         (q6-3) [above of=q6-1] {$q_6, 3, \downarrow$};
	\node[state]         (q6-5) [above of=q6-3] {$q_6, 5, \downarrow$};
	\node[state]         (q6-7) [above of=q6-5] {$q_6, 7, \equiv$};
	\node[state]         (q6-9) [above of=q6-7] {$q_6, 9, \equiv$};
	\node[state]         (q6-11) [above of=q6-9] {$q_6, 11, \equiv$};
	
	\node[state] 		 (q7-1) [below = 1.1cm of q4-2]  {$q_7, 1,\uparrow$};
	
	\node[state,accepting] 		 (q8-0) [below = 1.1cm of q7-1]  {$q_8, 0,\downarrow$};

	\path[->]
	(q0-0.east) edge node[near start] {\scriptsize $0$} (q1-1.west)
	(q0-3.east) edge node[near start] {\scriptsize $0$} (q1-4.west)
	(q0-6.east) edge node[near start] {\scriptsize $0$} (q1-7.west)
	(q0-9.east) edge node[near start] {\scriptsize $0$} (q1-10.west)
	
	(q1-1.west) edge[bend left = 20] node[near end,below] {\scriptsize $0$} (q2-2.east)
	(q1-4.west) edge[bend left = 20] node[near end,below] {\scriptsize $0$} (q2-5.east)
	(q1-7.west) edge[bend left = 20] node[near end,below] {\scriptsize $0$} (q2-8.east)
	(q1-10.west) edge[bend left = 20] node[near end,below] {\scriptsize $0$} (q2-11.east)
	
	(q1-1) edge[bend left = 10] node[near start,above] {\scriptsize $0$} (q3-2)
	(q1-4) edge[bend left = 10] node[near start,above] {\scriptsize $0$} (q3-5)
	(q1-7) edge[bend left = 10] node[near start,above] {\scriptsize $0$} (q3-8)
	(q1-10) edge[bend left = 10] node[near start,above] {\scriptsize $0$} (q3-11)

	(q1-1.south) edge[bend left = -20] node[near end,left] {\scriptsize $0$} (q4-2.west)
	(q1-4.east) edge[in = 150, out=20] node[near end,left] {\scriptsize $0$} (q4-5.west)
	(q1-10.north) edge[bend left = 80, in=90, out=90,looseness=1.95] node[near end,left] {\scriptsize $0$} (q4-11.north)
	
	(q2-2.west) edge[in = 180, out=180,looseness=1.3] node[near end] {\scriptsize $0$} (q0-3.west)
	(q2-5.west) edge[in = 180, out=180,looseness=1.3] node[near end] {\scriptsize $0$} (q0-6.west)
	(q2-8.west) edge[in = 180, out=180,looseness=1.3] node[near end] {\scriptsize $0$} (q0-9.west)
	(q2-11.north) edge[in = 230, out=110,looseness=1.3] node[near start, right] {\scriptsize $0$} (q0-6.west)
	
	(q3-2) edge[bend left = 10] node[near start,below] {\scriptsize $1$} (q1-1)
	(q3-5) edge[bend left = 10] node[near start,below] {\scriptsize $1$} (q1-4)
	(q3-8) edge[bend left = 10] node[near start,below] {\scriptsize $1$} (q1-7)
	(q3-11) edge[bend left = 10] node[near start,below] {\scriptsize $1$} (q1-10)
	
	(q4-2) edge node[right] {\scriptsize $1$} (q7-1)
	(q4-5) edge[bend left = 10] node[below] {\scriptsize $1$} (q5-4)
	(q4-11) edge[bend left = 10] node[below] {\scriptsize $1$} (q5-10)

	(q5-2) edge node[below] {\scriptsize $1$} (q6-1)
	(q5-4) edge node[below] {\scriptsize $1$} (q6-3)
	(q5-6) edge node[below] {\scriptsize $1$} (q6-5)
	(q5-8) edge node[below] {\scriptsize $1$} (q6-7)
	(q5-10) edge node[below] {\scriptsize $1$} (q6-9)
	
	(q5-6.west) edge[out=120,in=120,looseness=1.7] node[above] {\scriptsize $1$} (q6-11)
	
	(q6-1) edge node[below right] {\scriptsize $1$} (q5-0)
	(q6-3) edge node[below right] {\scriptsize $1$} (q5-2)
	(q6-5) edge node[below right] {\scriptsize $1$} (q5-4)
	(q6-7) edge node[below right] {\scriptsize $1$} (q5-6)
	(q6-9) edge node[below right] {\scriptsize $1$} (q5-8)
	(q6-11) edge node[below right] {\scriptsize $1$} (q5-10)

	(q7-1) edge node[right] {\scriptsize $1$} (q8-0)
	;
	\end{tikzpicture}
	\end{center}
	\caption{The automaton $M'$.}
	\label{fig-mprime}
	\end{figure}
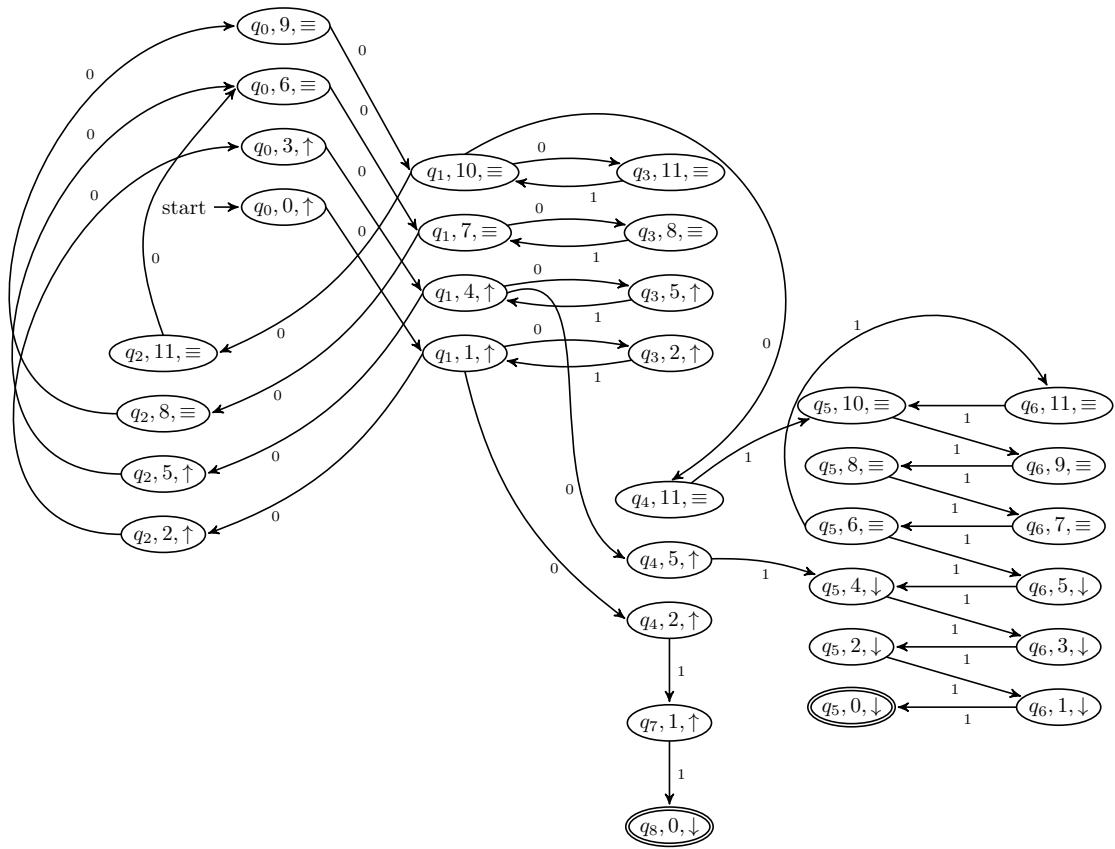
	
	\begin{proposition}
	\label{prop-consistent-runs}
		For each word $u=a_1\ldots a_n\in D_1$ and run $q_0\mathop{\longrightarrow}\limits^{a_1/R_1}q_1\mathop{\longrightarrow}\limits^{a_2/R_2}\ldots\mathop{\longrightarrow}\limits^{a_n/R_n}q_n$ in $M$ with $q_n\in F$ there is a run		
		$(q_0,0,\uparrow)\mathop{\longrightarrow}\limits^{a_1/R_1}(q_1,t_1,\sigma_1)\mathop{\longrightarrow}\limits^{a_2/R_2}\ldots\mathop{\longrightarrow}\limits^{a_n/R_n}(q_n,t_n,\sigma_n)$ in $M'$ with $(q_n,t_n,\sigma_n)\in F\times\{0\}\times\{\downarrow\}$ in $M'$.
	\end{proposition}
	\begin{proof}
		Let $u=a_1\ldots a_n\in D_1$ be a word and $q_0\mathop{\longrightarrow}\limits^{a_1/R_1}q_1\mathop{\longrightarrow}\limits^{a_2/R_2}\ldots\mathop{\longrightarrow}\limits^{a_n/R_n}q_n$ be a run in $M$ with $q_n\in F$.
		
		There are two cases: either $\mathrm{open}(v)<P$ for each prefix $v$ of $u$, or $\mathrm{open}(v)\geq P$ for
		at least one prefix $v$ of $u$. We construct an accepting run
		$(q_0,0,\uparrow)\mathop{\longrightarrow}\limits^{a_1/R_1}(q_1,t_1,\sigma_1)\mathop{\longrightarrow}\limits^{a_2/R_2}\ldots\mathop{\longrightarrow}\limits^{a_n/R_n}(q_n,t_n,\sigma_n)$
		of $M'$ in both cases.
		\begin{enumerate}
			\item If $\mathrm{open}(v)<P$ for each prefix $v$ of $u$, then let us define $t_i=\mathrm{open}(v)$ for each
			$0\leq i\leq n$, $\sigma_i=\uparrow$ for each $0\leq i<n$ and $\sigma_n=\downarrow$.		
			Then, the first $n-1$ transitions are of type i) and type ii) depending on $a_i$, with $\sigma_1=\sigma_2=\uparrow$, and the last transition is of type ii) with $\sigma_2=\downarrow$,
			since by $u\in D_1$ we get $a_n=1$. Thus this is indeed an accepting run in $M'$.
			\item If $\mathrm{open}(v)\geq P$ for at least one prefix $v$ of $u$, then let $i_\uparrow\geq 0$
			be the largest index
			so that for each $j\leq i_\uparrow$, $\mathrm{open}(a_1\ldots a_j)<P$ and let $i_\downarrow$ be
			the smallest index so that for each $j\geq i_\downarrow$, $\mathrm{open}(a_1\ldots a_j)<P$.
			These indices exist since $\mathrm{open}(a_1)=1<P$ and $\mathrm{open}(a_1\ldots a_n)=0<P$,
			moreover, $i_\uparrow<i_\downarrow$ since there exists some $i$ with $\mathrm{open}(a_1\ldots a_i)\geq P$
			and all of these $i$s have to fall strictly between $i_\uparrow$ and $i_\downarrow$.
			
			Now let us define
			\begin{align*}
			t_i&=\begin{cases}
			\mathrm{open}(a_1\ldots a_i)&\hbox{if }i\leq i_\uparrow\hbox{ or }i\geq i_\downarrow\\
			(\mathrm{open}(a_1\ldots a_i)~\mathrm{mod}~P)+P&\hbox{otherwise}
			\end{cases}&
			\sigma_i&=\begin{cases}
			\uparrow&\hbox{if }i\leq i_\uparrow\\
			\equiv&\hbox{if }i_\uparrow<i<i_\downarrow\\
			\downarrow&\hbox{if }i_\downarrow\leq i.
			\end{cases}
			\end{align*}
			We claim that for each $0\leq i<n$, $((q_i,t_i,\sigma_i),a_{i+1},(q_{i+1},t_{i+1},\sigma_{i+1}))$ is
			a transition in $M'$. Indeed: $(q_i,a_{i+1},q_{i+1})$ is a transition of $M$ and
			\begin{itemize}
				\item if $i<i_\uparrow$ and $a_{i+1}=0$, then $t_i=\mathrm{open}(a_1\ldots a_i)$, $t_{i+1}=\mathrm{open}(a_1\ldots a_{i+1})=t_i+1<P$ and $\sigma_1=\sigma_2=\uparrow$,
				thus then the triple is a type i) transition
				\item if $i<i_\uparrow$ and $a_{i+1}=1$, then $t_i=\mathrm{open}(a_1\ldots a_i)$, $t_{i+1}=\mathrm{open}(a_1\ldots a_{i+1})=t_i-1$, $t_i<P$ and $\sigma_1=\sigma_2=\uparrow$,
				thus then the triple is a type ii) transition
				\item if $i=i_\uparrow$, then (by the maximality of $i_\uparrow$) $a_{i+1}=0$,
				$\mathrm{open}(a_1\ldots a_i)=t_i=P-1$, $\mathrm{open}(a_1\ldots a_{i+1})=t_{i+1}=P$
				(as $(P~\mathrm{mod}~P)+P=0+P=P$), $\sigma_1=\uparrow$, $\sigma_2=\equiv$ and the triple is 
				a type iii) transition
				\item if $i_\uparrow<i<i_\downarrow-1$ and $a_{i+1}=0$, then $\sigma_i=\sigma_{i+1}=\equiv$,
				$t_i=(\mathrm{open}(a_1\ldots a_i)~\mathrm{mod}~P)+P\geq P$, $t_{i+1}=((\mathrm{open}(a_1\ldots a_i)+1)~\mathrm{mod}~P)+P\geq P$ and the triple is a type iii) transition
				\item if $i_\uparrow<i<i_\downarrow-1$ and $a_{i+1}=1$, then $\sigma_i=\sigma_{i+1}=\equiv$,
				$t_i=(\mathrm{open}(a_1\ldots a_i)~\mathrm{mod}~P)+P\geq P$, $t_{i+1}=((\mathrm{open}(a_1\ldots a_i)-1)~\mathrm{mod}~P)+P\geq P$ and the triple is a type iv) transition
				\item if $i=i_\downarrow-1$, then (by the minimality of $i_\downarrow$)
				$t_i=\mathrm{open}(a_1\ldots a_i)=P$, $a_{i+1}=1$, $t_{i+1}=\mathrm{open}(a_1\ldots a_{i+1})=P-1$,
				$\sigma_i=\equiv$, $\sigma_2=\downarrow$ and the triple is a type iv) transition
				\item if $i_\downarrow\leq i$ and $a_{i+1}=0$, then $t_i=\mathrm{open}(a_1\ldots a_i)$,
				$t_{i+1}=\mathrm{open}(a_1\ldots a_{i+1})=t_i+1<P$ and $\sigma_1=\sigma_2=\downarrow$,
				thus then the triple is a type i) transition
				\item if $i_\downarrow\leq i$ and $a_{i+1}=1$, then $t_i=\mathrm{open}(a_1\ldots a_i)$,
				$t_{i+1}=\mathrm{open}(a_1\ldots a_{i+1})=t_i-1$, $t_i<P$ and $\sigma_1=\sigma_2=\downarrow$,
				thus then the triple is a type ii) transition
			\end{itemize}
		\end{enumerate}
	\end{proof}
Let us call the run $r'$ of $M'$ constructed from a run $r$ of $M$ in the proof of Proposition~\ref{prop-consistent-runs}
the \emph{canonical lifted run} of $r$.
\begin{corollary}
	\label{cor-pm-is-pmprime}
	$L(M)=L(M')$ for the transducers $M$ and $M'$ of Proposition~\ref{prop-consistent-runs}.
	\end{corollary}
\begin{proof}

From Proposition~\ref{prop-consistent-runs} we have $L(M)\subseteq L(M')$.
For the other direction, $L(M')\subseteq L(M)$ also clearly holds
since the mapping $(q,n,\sigma)\mapsto q$ for each $q\in Q$, $n\in\tau(q)$,
$\sigma\in\{\uparrow,\downarrow,\equiv\}$ transforms an accepting run in $M'$
into an accepting run in $M$, with the same labels on the transitions.
\end{proof}

We will show that the created $M'$ has feasible cycles only. To this end, we first prove a batch of statements
regarding the states of $M'$.
\begin{proposition}
	\label{prop-props-of-states-of-mprime}
	The following all hold for the transducer $M'$ we constructed from $M$:
		\begin{enumerate}
	\item For any state $(q,n,\uparrow)$ of $M'$, (thus $0\leq n<P$), whenever
	$u\in\mathbf{Pref}(D_1)$ is a word with $(q_0,0,\uparrow)\mathop\leadsto\limits^u(q,n,\sigma)$,
	then $\mathrm{open}(u)=n$, moreover, at least one such word exists.
	\item For any $(q,n,\downarrow)$ of $M'$ (thus $0\leq n<P$), whenever
	$v\in\mathbf{Suf}(D_1)$ is a word with $(q,n,\sigma)\mathop\leadsto\limits^v(q_f,0,\downarrow)$ for
	some $q_f\in F$, then $\mathrm{close}(v)=n$, moreover, at least one such word exists.
	\item For any state $(q,n,\equiv)$ of $M'$ (and thus $P\leq n<2P$),
	whenever $u\in\mathbf{Pref}(D_1)$ is a word with $(q_0,0,\uparrow)\mathop\leadsto\limits^u(q,n,\equiv)$,
	then $\mathrm{open}(u)\equiv n~\mathrm{mod}~P$.
	\item For any state $(q,n,\equiv)$ of $M'$ (and thus $P\leq n<2P$),
	whenever $v\in\mathbf{Suf}(D_1)$ is a word with $(q,n,\equiv)\mathop\leadsto\limits^v(q_f,0,\downarrow)$
	for some $q_f\in F$, then $\mathrm{close}(v)\equiv n~\mathrm{mod}~P$.
	\item For any state $(q,n,\equiv)$ of $M'$ and integer $N\geq P$ with $N\equiv n~\mathrm{mod}~P$,
	there exists a word $uv\in D_1$
	such that $\mathrm{open}(u)=\mathrm{close}(v)=N$ and $(q_0,0,\uparrow)\mathop\leadsto\limits^u(q,n,\equiv)\mathop\leadsto\limits^v(q_f,0,\downarrow)$ for
	some $q_f\in F$.
\end{enumerate}
\end{proposition}
\begin{proof}
First observe that each transition either increases the second coordinate modulo $P$ when it's reading
a $0$, or decreases the second coordinate when it's reading a $1$. Thus in particular,
\begin{itemize}
	\item whenever $(q_0,0,\uparrow)\mathop\leadsto\limits^u(q,n,\sigma)$ for some word $u$ and
	state $(q,n,\sigma)$ of $M'$, it holds that $\mathrm{open}(u)\equiv n~\mathrm{mod}~P$,
	\item and whenever $(q,n,\sigma)\mathop\leadsto\limits^v(q_f,0,\downarrow)$ for some word $v$
	and state $(q,n,\sigma)$ of $M'$ and $q_f\in F$, it holds that $\mathrm{close}(v)\equiv n~\mathrm{mod}~P$.
\end{itemize}
These already prove Items $3$ and $4$ above. For $1$ and $2$, observe additionally that
\begin{itemize}
	\item whenever $(q_0,0,\uparrow)\mathop\leadsto\limits^u(q,n,\uparrow)$ for some state $(q,n,\uparrow)$
	of $M'$, then for each prefix $u'$ of $u$ we have $\mathrm{open}(u')<P$,
	\item and whenever $(q,n,\downarrow)\mathop\leadsto\limits^v(q_f,0,\downarrow)$ for some state
	$(q,n,\downarrow)$ of $M'$ and $q_f\in F$, then for each suffix $v'$ of $v$ we have $\mathrm{close}(v')<P$.
\end{itemize}
Indeed, for the first item to reach $(q,n,\uparrow)$ as there is no transition leading from a non-$\uparrow$
state to an $\uparrow$-state, we have to stay within the set of $\uparrow$-states during the whole run
reading $u$ starting from $(q_0,0,\uparrow)$. Then we can only use transitions of type i) and ii) and it
is easy to see by induction on $|u|$ that the first point holds, showing Item 1. A similar reasoning applies
to the second bullet point as once we are in a $\downarrow$-component, we cannot leave that, thus again
we can only use transitions of type i) and ii) starting from such a state, showing Item 3.

The parts ``moreover, at least such one such word exists'' parts are clear as during the construction
we explicitly remove all those states which are either not accessible or not coaccessible.

Now let us turn to Item 5 and let $(q,n,\equiv)$ be a state of $M'$ and let $N\geq P$ be an integer
with $N\equiv n~\mathrm{mod}~P$. Since by construction, we have the state $(q,n,\equiv)$ with $P\leq n<2P$
in $M'$ because $n\in\tau(q)$, and $n\geq P$, we get that $N\in N(q)$.
Thus, there exist words $u$ and $v$ with $uv\in D_1$ such that
$q_0\mathop\leadsto\limits^uq\mathop\leadsto\limits^vq_f$ for some $q_f\in F$ and
$\mathrm{open}(u)=\mathrm{close}(v)=N$.

Considering a run $r$ of $M$ on $uv$ which is in $q$ after reading in the $u$ prefix of the input,
let us see the canonical lifted run $r'$ of $r$ in $M'$. It has to be the case that in this run,
we have $(q_0,0,\uparrow)\mathop\leadsto\limits^u(q,m,\sigma)\mathop\leadsto\limits^v(q_f,0,\downarrow)$
for some $\sigma\in\{\uparrow,\downarrow,\equiv\}$ and value $m\equiv N~\mathrm{mod}~P$.
But as $\mathrm{open}(u)=N\geq P$, $\sigma$ cannot be $\uparrow$ due to Item 1;
as $\mathrm{close}(v)=N\geq P$, $\sigma$ cannot be $\downarrow$ due to Item 2; thus,
$\sigma=\equiv$ and hence $m=n$ (as that's the only possible value between $P$ and $2P-1$ inclusive
for which we have $m\equiv N~\mathrm{mod}~P$), proving Item 5.
\end{proof}	
Proposition~\ref{prop-props-of-states-of-mprime} has some interesting corollaries related to
having feasible cycles:
\begin{corollary}
\label{cor-cycles-of-m-prime}
The following all hold for the cycles present in $M'$:
\begin{enumerate}
	\item Whenever $(q,n,\sigma)\mathop\leadsto\limits^u(q,n,\sigma)$ is a cycle in $M'$ with $u\in\{0,1\}^+$,
	then any run in $M'$ corresponding the cycle either visits states only within a $\downarrow$-component,
	or within a $\uparrow$-component, or within a $\equiv$-component.
	\item Cycles within $\uparrow$- and $\downarrow$-components are always $0$-cycles.
	\item For any cycle $(q,n,\sigma)\mathop\leadsto\limits^w(q,n,\sigma)$ with $\sigma\in\{\uparrow,\downarrow\}$,
	and word $uv\in D_1$ with $q_0\mathop\leadsto\limits^uq\mathop\leadsto\limits^vq_f$ it holds that $uwv\in D_1$
	as well.
\end{enumerate}
\end{corollary}
\begin{proof}
	The first item is clear since a cycle can be present inside a component of $M'$ and we cannot reach an
	$\uparrow$-state from either a $\equiv$- or a $\downarrow$-state, nor an $\equiv$-state from a $\downarrow$-state
	so all the components of $M'$ are homogeneous with respect to $\sigma$.
	
	Assume $(q,n,\uparrow)\mathop\leadsto\limits^u(q,n,\uparrow)$ for some $q\in Q$ and $0\leq n<P$
	and let $x\in\{0,1\}^*$ be a word with $(q_0,0,\uparrow)\mathop\leadsto\limits^x(q,n,\uparrow)$.
	Then
	by Proposition~\ref{prop-props-of-states-of-mprime}, $\mathrm{open}(x)=n$. Now as
	$(q_0,0,\uparrow)\mathop\leadsto\limits^{xu}(q,n,\uparrow)$ also holds, we also know $\mathrm{open}(xu)=n$ as well,
	yielding $|u|_0=|u|_1$.
	
	For $\downarrow$-states having $0$-cycles only the proof is analogous.
	
	For the third point, if $(q,n,\sigma)\mathop\leadsto\limits^w(q,n,\sigma)$ with $\sigma\in\{\uparrow,\downarrow\}$,
	then in particular, during this cycle the run stays within the same $\sigma$-component all the time and the
	second coordinate always tracks the number of currently opened parentheses. Hence it cannot happen that
	for any prefix $w'$ of $w$ to have $\mathrm{close}(w')>n$ as then there would be an undefined transition
	from some state $(q,0,\sigma)$ with an input symbol $1$ during the run. Thus in that case, $uw\in\mathbf{Pref}(D_1)$ as well for any word $u$ with $(q_0,0,\uparrow)\mathop\leadsto^{u}(q,n,\sigma)$ (and at least one such word
	exists by Proposition~\ref{prop-props-of-states-of-mprime}) and from the second point, $\mathrm{open}(w)=0$
	so for any word $v$ with $(q,n,\sigma)\mathop\leadsto\limits^{v}(q_f,0,\downarrow)$ (and at least one
	such word exists) we get both $uv\in D_1$ and $uwv\in D_1$ as well.	
\end{proof}
Now we show the main result of this section:
\begin{proposition}
	The transducer $M'$ has feasible cycles only.
\end{proposition}
\begin{proof}
	For the first requirement of having feasible cycles only,
	let $(q,n,\sigma)$ be a state of $M'$ and $u\in\{0,1\}^+$ such that $(q,n,\sigma)\mathop\leadsto\limits^u(q,n,\sigma)$.
	\begin{itemize}
		\item If $\sigma\in\{\uparrow,\downarrow\}$, then Corollary~\ref{cor-cycles-of-m-prime} and Proposition~\ref{prop-props-of-states-of-mprime} show that the cycle can be extended into a run.
		\item If $\sigma=\equiv$, then let $N$ be the maximum value of $\mathrm{close}(u')$ ranging over
		the prefixes $u'$ of $u$. By Proposition~\ref{prop-props-of-states-of-mprime}, there exists some word
		$x\in\mathbf{Pref}(D_1)$ with $\mathrm{open}(x)\geq N+P$, $\mathrm{open}(x)\equiv n~\mathrm{mod}~P$
		and $(q_0,0,\uparrow)\mathop\leadsto\limits^x(q,n,\sigma)$. As $\mathrm{open}(x)$ is large enough,
		 $xu$ is then still in $\mathbf{Pref}(D_1)$ and $\mathrm{open}(xu)\geq P$, moreover, by
		 Proposition~\ref{prop-props-of-states-of-mprime}, $\mathrm{open}(xu)\equiv n~\mathrm{mod}~P$.
		 Again by Proposition~\ref{prop-props-of-states-of-mprime}, there exists then a word $v$ such 
		 that $v\in\mathbf{Suf}(D_1)$, $\mathrm{close}(v)=\mathrm{open}(xu)$ and $(q,n,\equiv)\mathop\leadsto\limits^{v}(q_f,0,\downarrow)$ showing the claim.
	\end{itemize}
	For the second requirement, observe that positive cycles can be present only in $\equiv$-components by
	Corollary~\ref{cor-cycles-of-m-prime}. Again, if $u$ is a positive cycle from the state
	$(q,n,\equiv)$, then we can construct a word $v$ with $(q_0,0,\uparrow)\mathop\leadsto\limits^{v}(q,n,\equiv)$ such that $\mathrm{open}(v)$ is large enough
	to make sure $vu$ is still in $\mathbf{Pref}(D_1)$. Then for each integer $t\geq 1$, the word $vu^t$
	still belongs to $\mathbf{Pref}(D_1)$ and by Proposition~\ref{prop-props-of-states-of-mprime}, to each such
	word there exists a suitable $w_t$ with $vu^tw_t\in D_1$ and $(q,n,\equiv)\mathop\leadsto\limits^{w_t}(q_f,0,\downarrow)$.
	
	For the last requirement, observe that the statement again requires for the state $(q,n,\sigma)$
	input words with arbitrary large $\mathrm{open}$ value to be completable from $(q,n,\sigma)$.
	By Proposition~\ref{prop-props-of-states-of-mprime}, this leaves only the possibility $\sigma=\equiv$.
	But repeating our previous argument, for any such cycle $u$ of $(q,n,\sigma)$ we indeed can pick
	some word $u_N\in\mathbf{Pref}(D_1)$ having a large enough opening value, leading to $(q,n,\sigma)$
	from the initial state, making sure that $u_Nu$ is still in $\mathbf{Pref}(D_1)$ and its opening is
	still at least $P$, hence there exists (again by Proposition~\ref{prop-props-of-states-of-mprime}) some
	suitable word $v$ with $u_Nuv\in D_1$, $(q,n,\sigma)\mathop\leadsto\limits^v(q_f,0,\downarrow)$.
\end{proof}

So we proved that each to and every transducer there exists an equivalent one which also satisfies the
feasible cycles property, which (along with Proposition~\ref{prop-ldelta-inductive} show that
for any scattered restricted one-counter language has a rank smaller than $\omega^2$, which,
applying the induction argument of Section 4, proves Theorem~\ref{thm-main}.

\section{Conclusion}

We confirmed the conjecture of~\cite{kuske} that scattered one-counter languages always have a rank
strictly smaller than $\omega^2$, thus in particular, well-ordered one-counter languages always
have an order type smaller than $\omega^{\omega^2}$. In the proof we used some upper bounds on the
rank -- it would be an interesting question to turn this into an algorithm which computes the exact
rank of the language. Also, since scattered order types lack a Cantor-like normal form, it is not
clear whether the order type of a scattered one-counter language is presentable by some expression
involving, say, $\omega$, $-\omega$, $1$, finite products, sums and powers and if so, whether such
a presentation is computable, or from the descriptive complexity point of view, whether representing
such an expression by a transducer can be more succint than storing the expression itself.
Also, it is still not known whether the order isomorphism problem of two scattered context-free
languages is decidable (for the general case of arbitrary context-free languages it is known to
be undecidable), and not even for one-counter languages. For the case of regular languages the
order isomorphism is known to be decidable, so to extend decidability the class of restricted
one-counter languages might be a good choice.

\section{Acknowledgements.}
This research was supported by project TKP2021-NVA-09. Project no. TKP2021-NVA-09 has been implemented with the support provided by the Ministry of Innovation and Technology of Hungary from the National Research, Development and Innovation Fund, financed under the TKP2021-NVA funding scheme.

The author wishes to thank two anonymous referees, sending valuable feedbacks to a much earlier version
of the manuscript, their inputs made it possible to improve the presentation of the result, very appreciated.

Also, thanks to Kitti Gelle for digitizing Figure~\ref{fig-mprime}.

%

\bibliography{biblio}{}
\bibliographystyle{plain}
\end{document}